\documentclass[a4paper, thm-restate,cleveref]{lipics-v2021}

\usepackage{graphicx} 
\usepackage{amssymb,amsmath,amsthm} 
\usepackage{todonotes} 
\usepackage{mathtools}
\usepackage{enumitem}
\usepackage{soul}
\usepackage{thmtools}
\usepackage{comment}
\usepackage{hyperref}
\usepackage{float}

\def\nth#1{#1^{\mbox{\tiny th}}}
\def\N{\mathbb{N}}
\def\ta{\mathtt{a}}
\def\tb{\mathtt{b}}

\def\tA{\mathtt{A}}
\def\tT{\mathtt{T}}
\def\tU{\mathtt{U}}
\def\tG{\mathtt{G}}
\def\tC{\mathtt{C}}

\newcounter{algno}

\DeclareMathOperator{\Fact}{Fact}

\DeclareMathOperator{\Pref}{Pref}

\DeclareMathOperator{\Suff}{Suff}

\newtheorem{problem}{Problem}
\newtheorem{question}[theorem]{Question}

\newcommand{\hpd}{%
\begin{tikzpicture}[scale=0.3]%
\draw (0,0) -- (1,0);%
\draw (0,0.2) -- (1,0.2);%
\draw (0.3,0) -- (0.3,0.2);%
\draw (0.6,0) -- (0.6,0.2);%
\draw (0.9,0) -- (0.9,0.2);%
\draw (1,0) arc (-150:150:0.2);%
\end{tikzpicture}
}

\newcommand{\hpdb}{\hpd}
\newcommand{\hpdbe}{\mathrel{\hpd}}

\newcommand{\hpdbpe}{\mathrel{\overset{\mathtt{p}}{\hpd}}}

\newcommand{\hpdbloge}{\mathrel{\overset{\mathtt{log}}{\hpd}}}

\newcommand{\hpdblogp}{\overset{\mathtt{p'log}}{\hpd}}
\newcommand{\hpdblogpe}{\mathrel{\overset{\mathtt{p'log}}{\hpd}}}
\newcommand{\hpdblogpmax}{\overset{\mathtt{p\uparrow}\text{-}\mathtt{log}}{\hpd}}
\newcommand{\hpdblogpmaxe}{\mathrel{\overset{\mathtt{p\uparrow}\text{-}\mathtt{log}}{\hpd}}}

\title{Programmable Co‑Transcriptional Splicing: Realizing Regular Languages via Hairpin Deletion}
\titlerunning{Realizing Regular Languages via Hairpin Deletion}

\author{Da-Jung Cho}{Department of Software and Computer Engineering, Ajou University, Republic of Korea}{dajungcho@ajou.ac.kr}{}{}
\author{Szil\'ard Zsolt Fazekas}{Graduate School of Engineering Science, Akita University, Japan}{szilard.fazekas@ie.akita-u.ac.jp}{}{}
\author{Shinnosuke Seki}{University of Electro-Communications, Tokyo, Japan}{s.seki@uec.ac.jp}{}{} 
\author{Max Wiedenh\"oft}{Department of Computer Science, Kiel University, Germany}{maw@informatik.uni-kiel.de}{}{}

\authorrunning{D.\,-J. Cho, S.\,Z. Fazekas, S. Seki, and M. Wiedenh\"oft}
\Copyright{Da-Jung Cho, Szil\'ard Z. Fazekas, Shinnosuke Seki, and Max Wiedenh\"oft}

\ccsdesc[300]{Applied computing~Bioinformatics}
\keywords{RNA Transcription, Co-Transcriptional Splicing, Finite Automata Simulation, NFA Minimization}
\date{}

\funding{\emph{Da-Jung Cho} was supported by Institute of Information \& communications Technology Planning \& Evaluation~(IITP) under the Artificial Intelligence Convergence Innovation Human Resources Development~(IITP-2025-RS-2023-00255968) grant funded by the Korea government~(MSIT). \emph{Szil\'ard Zsolt Fazekas} was supported by JSPS Kakenhi Grant No. 23K10976. \emph{Max Wiedenhöft's} work was partially supported by the DFG project number 437493335.}

\hideLIPIcs
\nolinenumbers

\EventEditors{Josie Schaeffer and Fei Zhang}
\EventNoEds{2}
\EventLongTitle{31st International Conference on DNA Computing and Molecular Programming (DNA 31)}
\EventShortTitle{DNA 31}
\EventAcronym{DNA}
\EventYear{31}
\EventDate{August 25--29, 2025}
\EventLocation{Lyon, France}
\EventLogo{}
\SeriesVolume{347}
\ArticleNo{5}

\begin{document}

\maketitle
\begin{abstract}
RNA co-transcriptionality, where RNA is spliced or folded during transcription from DNA templates, offers promising potential for molecular programming. It enables programmable folding of nano-scale RNA structures and has recently been shown to be Turing universal. While post-transcriptional splicing is well studied, co-transcriptional splicing is gaining attention for its efficiency, though its unpredictability still remains a challenge.
In this paper, we focus on engineering co-transcriptional splicing, not only as a natural phenomenon but as a programmable mechanism for generating specific RNA target sequences from DNA templates. The problem we address is whether we can encode a set of RNA sequences for a given system onto a DNA template word, ensuring that all the sequences are generated through co-transcriptional splicing. Given that finding the optimal encoding has been shown to be NP-complete under the
various energy models considered~\cite{ChoFSM25}, we propose a practical alternative approach under the logarithmic
energy model. 
More specifically, we provide a construction that encodes an arbitrary nondeterministic finite automaton (NFA) into a circular DNA template from which co-transcriptional splicing produces all sequences accepted by the NFA. As all finite languages can be efficiently encoded as NFA, this framework solves the problem of finding small DNA templates for arbitrary target sets of RNA sequences. The quest to obtain the \emph{smallest} possible such templates naturally leads us to consider the problem of minimizing NFA and certain practically motivated variants of it, but as we show, those minimization problems are computationally intractable.
\end{abstract}

\newpage
\section{Introduction}

RNA splicing is a fundamental process in eukaryotic gene expression, where intronic sequences are removed from precursor messenger RNA (pre-mRNA) transcripts, and the remaining exonic sequences are ligated together to form a mature messenger RNA~(mRNA).
This essential process is mediated by the spliceosome, a dynamic and large macromolecular complex consisting of small nuclear RNAs~(snRNAs) and associated proteins.
The spliceosome assembles on the pre-mRNA in a highly conserved and sequential manner. The assembly involves the sequential recruitment of specific snRNPs~(small nuclear ribonucleoproteins) and various factors that play critical roles in splice site recognition and catalysis. 
The process begins when the 5'-splice site~(5'-SS) is identified by the polymerase-spliceosome complex. This site is then kept in close proximity to the transcribed region of the RNA, awaiting the hybridization event with the complementary sequence at the 3'-splice site~(3'-SS) after subsequent transcription of the RNA. Once the 3'-SS is transcribed, the intronic sequence is excised, and the two exons are ligated together to form the mature RNA. Splicing takes place while transcription is still ongoing, meaning that the spliceosome assembles and acts on the RNA as it is being made. This close coordination between transcription and splicing is known as \emph{co-transcriptional splicing}.  Increasing biochemical and genomic evidence has established that spliceosome assembly and even catalytic steps of splicing frequently occur co-transcriptionally~\cite{MerkhoferHJ14}. As shown in Figure~\ref{fig:CTS}, co-transcriptional splicing involves the spliceosome's early recognition and action on splice sites while the RNA is still being transcribed. This coupling imposes temporal and structural constraints on splice site selection and nascent RNA folding, since spliceosome assembly and splicing catalysis occur during ongoing transcription.

\begin{figure}[h!]
    \centering
    \includegraphics[scale=0.32]{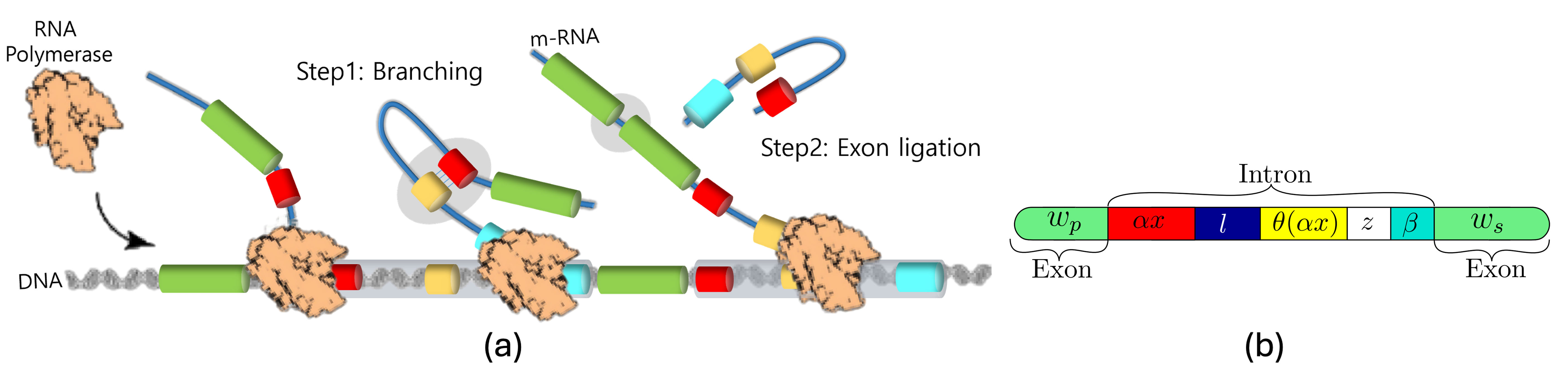}
    \caption{(a) An illustration of co-transcriptional splicing: The process begins when the 5'-splice site~(5'-SS) is identified by the polymerase-spliceosome complex. This site is then kept in close proximity to the transcribed region of the RNA, awaiting the hybridization event with the complementary sequence at the 3'-splice site~(3'-SS) after subsequent transcription of the RNA. The transcript forms a hairpin where the 5'-SS binds the branch point. After the 3'-SS is transcribed the spliceosome removes the intron and joins the exons. (b) A string representation of the DNA sequence in~(a), factorized according to the scheme of co-transcriptional splicing, as used in Definition~\ref{def:operation}.}
    \label{fig:CTS}
\end{figure}

While co-transcriptional splicing has been studied as a natural phenomenon, it is increasingly being recognized as a programmable and engineerable process. 
Indeed, Geary, Rothemund, and Andersen have proved that one of such processes called co-transcriptional folding is programmable to assemble nano-scale single-stranded RNA structures in vitro~\cite{GearyRA14} and then Seki, one of the authors, proved in collaboration with Geary and others that it is Turing universal by using an oritatami model~\cite{GearyMSS19}. 
Recently, we became curious about the potential of using RNA sequences as programmable components in molecular systems. Motivated by advances in co-transcriptional RNA folding and splicing, this raises the question: Can we encode a set of RNA sequences for a given system onto a single DNA template such that all the sequences are generated through co-transcriptional splicing, while avoiding sequences that may disrupt the system? 
As transcription proceeds and the 3'-SS is transcribed, RNA polymerase often pauses just downstream, allowing time for accurate splice site recognition and splicing to take place. The folding of the emerging RNA strand can influence this process. 
Motivated by this perspective, we introduced a formal model of co-transcriptional splicing~\cite{ChoFSM25}, defined as \emph{contextual lariat deletion} under various biologically plausible conditions (energy models). For the sake of simplicity and to avoid confusion with the rather different but similarly named \emph{contextual deletion}~\cite{KariT96}, we will refer to the operation here as \emph{hairpin deletion}. 
We showed that the template construction problem is NP-complete when $|\Sigma| \geq 4$.
The hardness of this problem led us to consider an alternative approach for the template construction problem. A natural observation is that the DNA template word can be obtained by finding the shortest common supersequence of the RNA target sequences. However, this problem is also known to be NP-complete~\cite{GareyJ02, Maier78}. 
The contributions of this paper are as follows:
\begin{itemize}
\item \textbf{(Parallel) Logarithmically-bounded hairpin deletion operation}: In natural RNA folding, the destabilizing energy contributed by loop regions in hairpin structures increases logarithmically with loop size. While small loops incur significant energy penalties, this effect diminishes as loop length increases~\cite{EinertN11}. This behavior supports the logarithmic-bounded hairpin model, where the destabilization caused by a loop of length $\ell$ is proportional to $\log(\ell)$. Our model exploits this property by favoring hairpin structures with long loops and short stems, which remain energetically viable under the logarithmic cost function.
Furthermore, we introduce a parallel deletion model to represent this consecutive splicing behavior since co-transcriptional splicing tends to proceed consecutively rather than recursively due to kinetic and structural constraints during transcription.
\item \textbf{Template word construction by encoding RNA target sequences to a finite automata on a circular DNA template word}: In Section~\ref{sec_sim}, our main result shows, if we consider a set of target RNA sequences as a regular language represented by some (non)deterministic finite automaton, we can construct a DNA template from which we can obtain exactly that language using the logarithmic hairpin deletion model. 
The construction operates on a circular DNA template, which is widely used in molecular biology, for instance in plasmids and viral genomes~\cite{DelGRED98, GearyRA14}, and supports continuous transcription.
We first demonstrate how to construct such a circular DNA word that encodes an arbitrary NFA according with the behavior of logarithmic hairpin deletion.
Then, a detailed construction for the RNA alphabet~$\Sigma_{\mathtt{RNA}} = \{\tA,\tU,\tC,\tG\}$ is also provided, illustrating how the states and transitions of the automaton are encoded. It is shown that this construction actually works for an appropriately designed DNA template word under the logarithmic hairpin deletion model.
\item \textbf{Minimizing NFA-like models}: As the size of the constructed template depends on the number of states and transitions in the automaton, minimizing the NFA becomes an important consideration. As the minimization problem for NFAs is NP-hard in general~\cite{JiangR93}, in Section~\ref{sec_dec}, we consider restricted NFA-like models that fit our setting. We show that minimization in these restricted cases is still computationally hard.
This highlights a fundamental complexity barrier in optimizing DNA template designs for generating arbitrary regular languages of RNA sequences via co-transcriptional splicing. Therefore, it may be more effective to focus on designing specific classes of target sequences and their corresponding automata representations.
\end{itemize} 

\newpage

\section{Preliminaries}
\label{sec_prel}
Let $\N$ denote the set of positive integers and let $\N_0 = \N\cup\{0\}$. 
Let $\mathbb{Z}$ denote the set of integers. For some $m\in\N$, denote by $[m]$ the set $\{1,2,...,m\}$ and let $[m]_0$ = $[m]\cup\{0\}$. With $\Sigma$, we denote a finite set of symbols, called an \emph{alphabet}. The elements of $\Sigma$ are called \emph{letters}. 
A \emph{word} over $\Sigma$ is a finite sequence of letters from $\Sigma$. 
With $\varepsilon$, we denote the \emph{empty word}. 
The set of all words over $\Sigma$ is denoted by $\Sigma^*$. Additionally, set $\Sigma^+ = \Sigma^*\setminus\{\varepsilon\}$.
Given some word $w\in\Sigma^*$, the word $w^\omega$ represents an infinite repetition of $w$, called \emph{circular word}.
The \emph{length} of a word $w\in\Sigma^*$, i.e., the number letters in $w$, is denoted by $|w|$; hence, $|\varepsilon| = 0$. 
For some $w\in\Sigma^*$, if we have $w = xyz$ for some $x,y,z\in\Sigma^*$, then we say that $x$ is a \emph{prefix}, $y$ is a \emph{factor}, and $z$ is a \emph{suffix} from $w$. We denote the set of all factors, prefixes, and suffixes of $w$ by $\Fact(w)$, $\Pref(w)$, and $\Suff(w)$ respectively. If $x\neq w$, $y\neq w$, or $z\neq w$ respectively, we call the respective prefix, factor, or suffix \emph{proper}. For some word $w = xy$, for $x,y\in\Sigma^*$, we define $w\cdot y^{-1} = x$ as well as $x^{-1}\cdot w = y$. A function $\theta: \Sigma^* \to \Sigma^*$ is called an \textit{antimorphic involution} if $\theta(\theta(\ta)) = \ta$ for all $\ta \in \Sigma$ and $\theta(w\tb) = \theta(\tb)\theta(w)$ for all $w \in \Sigma^*$ and $\tb \in \Sigma$. Watson-Crick complementarity, which primarily governs hybridization among DNA and RNA sequences, has been modeled as an antimorphic involution that maps {\tt A} to {\tt T} ({\tt U}), {\tt C} to {\tt G}, {\tt G} to {\tt C}, and {\tt T} ({\tt U}) to {\tt A}.  A nondeterministic finite automaton~(NFA) is a tuple
$A = (\Sigma, Q, q_0,\delta, F)$ where $\Sigma$ is
the input alphabet, $Q$ is the finite set of states,
$\delta \colon Q \times \Sigma \rightarrow 2^Q$ is
the multivalued transition function, $q_0 \in Q$ is the
initial state, and $F \subseteq Q$ is the set of final states.
In the usual way, $\delta$ is extended as a function
$Q \times \Sigma^* \rightarrow 2^Q$ and the language accepted
by $A$ is $L(A) = \{ w \in \Sigma^* \mid \delta(q_0, w) \cap F
\neq \emptyset \}$. The automaton $A$ is a deterministic finite
automaton~(DFA) if $\delta$ is a single valued partial function.
It is well known that deterministic and nondeterministic
finite automata recognize the class of {\em regular languages}~\cite{hopcroft}. Finally, for a regular expression $r$, writing it down, immediately refers to its language, e.g., writing $\ta(\ta|\tb)^*\tb$ refers to the set $\{w\in\{\ta,\tb\}^*\mid w[1] = \ta, w[|w|] = \tb\}$.

    \subsection{Hairpin deletion}

Hairpin deletion is a formal operation introduced in~\cite{ChoFSM25} with the purpose of modeling co-transcriptional splicing in RNA transcription. The computational power of the operation in various energy models (length restrictions between the stem and the loop) has been thoroughly investigated in~\cite{ChoFSM25}. Here we introduce only the elements relevant to our results (in the \emph{logarithmic energy model}).
A pair of words, $(u, v) \in \Sigma^* \times \Sigma^*$, serves as a context; $u$ and $v$ are particularly called \textit{left} and \textit{right} contexts. 
A set $C$ of such pairs is called a \textit{context set}. Hairpins are an atom of DNA and RNA structures. They can be modeled as $x \ell \theta(x)$ for words $x, \ell \in \Sigma^*$, where $x = b_1b_2 \cdots b_n$ and its Watson-Crick complement $\theta(x) = \theta(b_n) \cdots \theta(b_2)\theta(b_1)$ hybridize with each other via hydrogen bonds between bases $b_i$ and $\theta(b_i)$ into a stem, and $\ell$ serves as a loop (in reality $|\ell| \ge 3$ is required) \cite{KariLKST06}. Assume $\Sigma$ to be some alphabet, $w\in\Sigma^*$ to be some word over that alphabet, and $\theta$ to be some antimorphic involution on $\Sigma$. We call $w$ a \emph{hairpin} if $w = x\ell\theta(x)$, for some $x,\ell\in\Sigma^*$. Furthermore, let $c\in\N$ be some positive number. We call $w$ a \emph{logarithmically-bounded hairpin (log-hairpin)} if $|x| \geq c\cdot\log(\ell)$. We denote the set of all hairpins regarding $\Sigma$ and $\theta$ by $H(\Sigma,\theta)$ and the sets of all log-hairpins regarding $\Sigma$, $\theta$, and $c$ by $H_{log}(\Sigma,\theta,c)$. For readability purposes, we often write just $H$ or $H_{log}$ and infer $\Sigma$, $\theta$ and $c$ by the context. Notice that the logarithmic penalty of $\ell$ in log-hairpins allows for an exponential size of the loop with respect to the length of the stem.

 Using the notion of hairpins and log-hairpins, we can continue with the definition of bounded hairpin deletion. Splicing relies on the occurrence of left and right contexts, the formation of a stable hairpin containing the left context, and a potential margin between the hairpin and the occurrence of the right context. As only the logarithmic energy model is relevant in this paper, we focus only on log-hairpins.
 Figure~\ref{fig:CTS}(b) shows a string factorization where the logarithmic bounded hairpin deletion operation can be applied, reflecting the co-transcriptional splicing process.

\begin{definition} \label{def:operation}
    Let $S = (\Sigma,\theta,c,C,n)$ be a tuple of parameters, with $\Sigma$ an alphabet, $w\in\Sigma$ a word, $C$ a context-set, $\theta$ an antimorphic involution, $n\in\N$ a constant called \emph{margin}, and $c\in\N$ a multiplicative factor used in the logarithmic energy model definition. If $$w = w_p\ \alpha\ x\ \ell\ \theta(x)\theta(\alpha)\ z\ \beta\ w_s$$ for some $(\alpha,\beta)\in C$ and $w_p,w_s,x,\ell,z\in\Sigma^*$, then we say that $w_pw_s$ is \emph{obtainable by logarithmic bounded hairpin deletion} (or just \emph{log-hairpin deletion}) \emph{from $w$ over $S$} (denoted $w \hpdbloge_{S} w_pw_s,$) if and only if $\alpha x \ell \theta(x\alpha)\in H_{log}$ and $|z| \leq n$.
\end{definition}

For example, if $\theta$ is defined to represent the Watson-Crick complement, the context-set is set to $C = \{({\color{red}\tA\tU\tA},{\color{blue}\tC\tC\tC})\}$, the {\color{purple}margin} is set to $n=1$, the log-factor is set to $c=1$, and $w = \tA\tA\tU\tA\tA\tC\tC\tU\tU\tA\tU\tG\tC\tC\tC\tG\tA$, then we have
$$ \tA{\color{red}\tA\tU\tA}{\color{orange}\tA}\tC\tC{\color{orange}\tU\tU\tA\tU}{\color{purple}\tG}{\color{blue}\tC\tC\tC}\tG\tA \hpdbloge \tA\tG\tA $$
by $w_p = \tA$, $\alpha = {\color{red}\tA\tU\tA}$, $x = {\color{orange}\tA}$, $\ell = \tC\tC$, $\theta(x)\theta(\alpha)={\color{orange}\tU\tU\tA\tU}$, $z = {\color{purple}\tG}$, $\beta = {\color{blue}\tC\tC\tC}$, and $w_s = \tG\tA$.

Given some input word or some input language, we can consider the language of all words obtained if hairpin deletion is allowed to be applied multiple times. In practice, it is observed that co-transcriptional splicing occurs in a consecutive manner on the processed DNA sequence, meaning that newly created splicing sites resulting from earlier splicing operations do not affect the outcome. Hence, we introduce a parallel deletion model in which only non-overlapping hairpins present in the sequence at the beginning can be deleted, representing consecutive co-transcriptional splicing.

\begin{definition}
    Let $w\in\Sigma^*$ be a word and let $S$ be a tuple of parameters for log-hairpin deletion. For some $w'\in\Sigma^*$, we call it \emph{obtainable by parallel log-hairpin deletion of $w$ over $S$}, denoted $w \hpdblogpe_S w'$, if there exist $m\in\N$ and $u_1,\alpha_1,...,u_{m},\alpha_m,u_{m+1}\in\Sigma^*$ such that
    $$w = u_1\alpha_1u_2\alpha_2...u_m\alpha_mu_{m+1},$$
    $$w' = u_1u_2...u_mu_{m+1},$$
    and, for all $i\in[m]$, we have $\alpha_i\hpdbloge_S\varepsilon$, i.e., each $\alpha_i$ is removed by log-hairpin deletion. For readability purposes, we may just write $\hpdb$ and infer $S$ or $\mathtt{log}$ by context.
\end{definition}

For example, given the context set $C=\{({\color{red}\tA\tA},{\color{red}\tC\tC}),({\color{blue}\tC\tC},{\color{blue}\tC\tC})\}$, $\theta$ defined by $\theta(\tA) = \tU$ and $\theta(\tC) = \tG$, and the word $w = \tC\tA\tG\tA\tG\tU\tC\tU\tG\tC\tC\tA\tA\tG\tG\tG$, we could delete two factors at once:
$$\tC\tA\ 
{\color{red}\tA\tA}\tG{\color{orange}\tU\tU}{\color{red}\tC\tC}\ 
\tU\tG\ 
{\color{blue}\tC\tC}\tA\tA{\color{cyan}\tG\tG}{\color{blue}\tC\tC}\ 
\tG \hpdbpe \tC\tA\tU\tG\tG$$

Applying parallel log-hairpin deletion to a word or a given language can produce another language. Hence, given some $w\in\Sigma^*$ or $L\subseteq\Sigma^*$, the \emph{parallel log-hairpin deletion sets over $S$ of $w$ (or $L$)} are defined by $$[w]_{\hpdblogp_S} = \{\ w'\in\Sigma^*\mid w\hpdbloge_S w'\ \}\text{ (or }[L]_{\hpdblogp} = \bigcup_{w\in L}[w]_{\hpdblogp}\text{).}$$

In practice and for the purpose of the following constructions in the main body of the paper, we define a variant of the parallel hairpin deletion that uses certain necessary assumptions about the words contained in the language. First of all, if there exist isolated 5'SS (left contexts) in a word that is read from left to right, there is a high chance that it is used in the formation of the hairpin if a corresponding 3'SS exists~\cite{BeyerO88,MerkhoferHT14}. In that sense, we assume some greediness when it comes to selecting left contexts. Similarly, we could assume the same about the choice of the right context. For the construction that follows, only the first assumption is necessary. If there exist overlapping 5'SS in the word, thermodynamics might result in some nondeterministic choice of the actual 5'SS that is used~\cite{OreillyNB95, RocaSK05}. To obtain a model that follows these constraints, we introduce a variant of parallel deletion, called maximally parallel deletion, in which it is assumed that no left contexts survive in each part $u_i$.

\begin{definition}
    Let $w\in\Sigma^*$ be a word and $S$ be some tuple of parameters for log-hairpin deletion. For some $w'\in\Sigma^*$, we call it obtainable by \emph{maximally parallel log-hairpin deletion over $S$ of $w$}, denoted $w \hpdblogpmaxe_S w'$, if there exists some $m\in\N$ and $u_i,\alpha_i,u_{m+1}\in\Sigma^*$, $i\in[m]$, such that $w = u_1\alpha_1...u_n\alpha_nu_{n+1}$, $w' = u_1...u_{n+1}$, bounded hairpin deletion cannot be applied to $u_{n+1}$, and, for all $i\in[m]$, we have $\alpha_i \hpdbloge_S \varepsilon$ as well as, for all $(x,y)\in C$, we have $x\notin\Fact(u_i)$. 
\end{definition}

Again, we denote the set of all words obtainable by maximally parallel log-hairpin deletion by $[w]_{\hpdblogpmax_{S}}$ (analogously for input languages as before). Notice that $[w]_{\hpdblogpmax_{S}}\subseteq [w]_{\hpdblogp_{S}}$ as it is a more restricted variant of the parallel log-hairpin deletion set. This concludes all necessary introductory terminology regarding the formal model of hairpin deletion.

\section{Simulating Finite Automata with Maximal Parallel Log-Hairpin Deletion}
\label{sec_sim}
This section investigates the possibility to obtain arbitrary regular languages of RNA sequences from a circular template DNA sequence using bounded hairpin deletion. We provide a construction that allows for the simulation of arbitrary (non)deterministic finite automata (DFA/NFA) using maximally parallel bounded hairpin deletion in the logarithmic energy model on circular DNA. In particular it is shown, given some NFA $A$, that we can construct a word $w$ for which we have $[w^\omega]_{\hpdblogpmax} = L(A)$.

As an initial idea, each transition defined in some NFA $A$ was encoded on a circular word in a consecutive matter, i.e., if for example $A$ had the transitions $(q_i,\ta,q_j)$ and $(q_j,\tb,q_\ell)$, then $w$ would contain them directly as factors, resulting in a word $w = \cdots  (q_i,\ta,q_j) \cdots  (q_j,\tb,q_\ell) \cdots $. A context-set $C$ can be defined to contain some context $(\alpha,\beta)\in C$ with $\alpha = ,q_j)$ and $\beta=(q_j,$ that would allow for the potential deletion of the factor $,q_j)\cdots (q_j,$ and resulting in a word $w' = \cdots  (q_i,\ta\tb,q_\ell$, allowing for further parallel deletions. This model works quite intuitively, but resulted in various problems. For example, the controlled formation of hairpins with a stem and a loop posed a serious challenge. Also a controlled notion of termination was missing. Due to the above, we decided on a different approach. Now, instead of encoding all transitions in a parallel manner, we use a single factor $s_i$ per state $q_i$ that allows a non-deterministic transition to jump to some other factor $s_j$, encoding $q_j$, using hairpin deletion. The general spirit, however, stays the same, as we are still moving around a circular DNA template and select transitions to be taken non-deterministically by the application of hairpin deletion, leaving only the letters labeling those transitions to remain in the resulting words.

First, we need some section-specific definitions, as we are now considering and utilizing infinite repetitions of some circular DNA template. As mentioned in Section~\ref{sec_prel}, a circular word $w^\omega$ is an infinite repetition of some word $w\in\Sigma^*$. In order to extend the notion of hairpin deletion to infinite words, we need some notion of intentionally stopping the transcription process on infinite words.

In practice there are, among others, two different ways to stop the transcription process and detach the produced RNA from the polymerase. First, there is rho-dependent transcription termination which uses other molecules that bind to certain factors on the produced RNA~\cite{StewartLY86}. Second, there is rho-independent transcription termination that is based on the formation of a hairpin of certain length followed by a specific factor on the template DNA~\cite{CarafaBT90}. The rho-dependent model would be easier to embed in our definition, but it relies on external factors to work. Hence, to obtain a process that works as independent from external factors as possible, we will use the latter, rho-independent, model. 

Most of the time, after the final hairpin, a factor containing only a certain number of $\tU$'s/$\tT$'s is enough to result in the decoupling of the produced RNA sequence from the polymerase. To obtain a general model, we allow for arbitrary words to be used at this point, elements of some $T\subset\Sigma^*$, called the \emph{set of terminating-sequences}. We extend the notion of log-hairpin deletion by adding $T$ and a \emph{terminating stem-length} $m\in\N$ to the properties $S$, i.e., writing $S = (\Sigma,\theta,c,C,n,T,m)$ instead of $S=(\Sigma,\theta,c,C,n)$, if the processed word is a circular word $w^\omega$ over $w\in\Sigma^*$. The number $m$ represents the minimal length of the hairpin or stem that must be formed as a suffix in the RNA produced, before reading a terminating-sequence $t\in T$ in $w^\omega$. Using this, we can adapt the notion of log-hairpin deletion for circular words by adding the notion of termination.

\begin{definition}\label{def:circ-word-hpd}
    Let $w^\omega$ be a circular word over $w\in\Sigma^*$. Let $S=(\Sigma,\theta,c,C,n,T,m)$ be a tuple of properties for log-hairpin deletion as defined before. Assume there exists a finite prefix $w't$ of $w^\omega$ with $w'\in\Sigma^*$ and $t\in T$ such that $w' \hpdbe_S u$, for some word $u\in\Sigma^*$.
    If $u$ has a suffix $s$ of length $|s| = 2m$ such that $s$ forms a stem\footnote{The requirement that $s$ is only a stem of length $2m$ and not a hairpin with a stem of such length, i.e., no loop $\ell$ is formed, is a technical one chosen for simplicity reasons. A long enough stem always allows for a loop to be formed. Hence, adapting the encoded suffix responsible for transcription termination in the following construction in the proof of Theorem~\ref{theorem:dfa-simulation} should also be possible by setting $m$ short enough and adding enough letters to that suffix that do not interfere with the constructed context-set}, i.e., $s = x\theta(x)$, for some $x\in\Sigma^*$, then we say that $u\cdot s^{-1}$ is obtainable from $w^\omega$ by \emph{terminating log-hairpin deletion}.
\end{definition}

The remaining part of this section provides a general framework for constructing working DNA templates. Properties needed for the construction to work are provided. The following result is the main result of this section and contains the general construction of a template word $w$ that allows for the simulation of arbitrary regular languages represented by finite automata using terminating maximally parallel bounded log-hairpin deletion. We provide a basic construction for NFAs which naturally follows for DFAs and potentially other finite automata models such as GNFAs or GDFAs (which allow sequences as transition labels instead of just single letters).

\begin{theorem}\label{theorem:dfa-simulation}
    Let $A = (Q,\Sigma,q_1,\delta,F)$ be some NFA. There exists a word $w\in\Sigma^*$ and a properties tuple for log-hairpin deletion $S = (\Sigma,\theta,c,C,n,T,m)$ as defined before such that $L(A) = [w^\omega]_{\hpdblogpmax_{S}}$.
\end{theorem}
\begin{proof}
    Let $A = (Q,\Sigma,q_1,\delta,F)$ be some NFA with $Q = \{q_1,\dots ,q_o\}$, for some $o\in\N$, and $\Sigma = \{\ta_1,\dots ,\ta_\sigma\}$, for some $\sigma\in\N$. We begin by an explanation of the basic idea. Then, we continue with a formal construction which is then given an intuitive explanation that utilizes images afterwards. Due to space constraints, the full proof of the correctness of the construction is given in Appendix~\ref{section:appendix-simulation-proof}. Additionally, an example of an actual implementation using the RNA alphabet $\Sigma_{\mathtt{RNA}} = \{\tA,\tU,\tC,\tG\}$ for a given NFA is provided in Appendix~\ref{sec:appendix-simulation-example-specific}.

    \textbf{Basic idea:} For each state $q_i\in Q$, we construct a word $s_i\in\Sigma^*$ that represents this state and the outgoing transitions from it, as defined by $\delta$. Each period $w\in\Sigma^*$ of the circular word $w^\omega$, over which transcription is done, will be essentially made up of a concatenation of all $s_i$'s with the addition of one final word $s_e$, i.e., we obtain $w = s_1s_2\cdots s_os_e$ (see Figure~\ref{fig:simulation-circular-transcription}). The suffix $s_e$ handles transcription termination. Hairpin deletion will be used to simulate a transition between two connected states $q_i$ and $q_j$, while reading a letter $\ta\in\Sigma$. This is done by jumping from $s_i$ to $s_j$ by removing everything but the letter $\ta$ between $s_i$ and $s_j$. This is implemented using a two-step process for which context-sets will be deliberately designed such that these hairpin deletion steps actually simulate the process of reading letters in the automaton $A$. 
    To terminate transcription, in final-states, a technical transition to the end of the template can be used that jumps to a factor $t\in T$ at the end of $w$ while obtaining a suffix which consists of a stem of size $2m$ at the end of the transcribed word. By that, we successively build prefixes of words in $L(A)$ and may finalize the transcription process anytime a final state is reached, effectively obtaining only words in $L(A)$. 

    \begin{figure}[h]
    \centering
    \includegraphics[width=5.25cm]{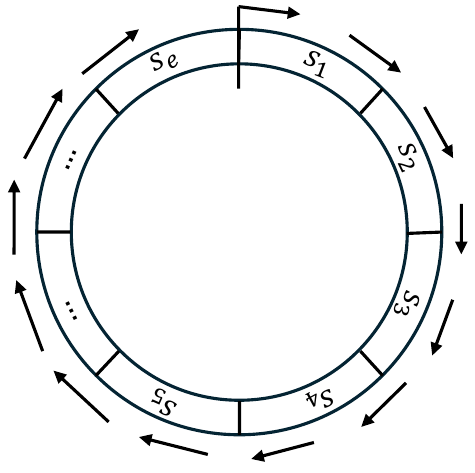}
    \caption{Representation of circular transcription (maximally parallel hairpin deletion over $w^\omega$).}
    \label{fig:simulation-circular-transcription}
    \end{figure}

    \textbf{Construction:} We continue with the formal construction of $w$. After that, we give an intuitive sketch of the functionality.\footnote{Consider reading the intuition-part in parallel to the construction-part for a better understanding.} We need some general assumptions. $\theta$ is not defined specifically in the general case, but rather implied by the constraints we give in this proof. The logarithmic factor $c$ is set to $1$ in this construction, but larger numbers also work. The margin number $n$ is set to $0$ in this construction. Such a constant may be added but is not necessary for this construction to work. $\Sigma$ is copied from the definition of $A$, but a distinct alphabet may simplify the construction. For practical implementation purposes, we use this assumption to keep the possible alphabet size as low as possible. Consider the construction in the Appendix to see an actual working example of an encoding for an alphabet size of $4$. $T$ and $m$ are not explicitly specified, but construction constraints are given.

    In the final construction, we obtain a word $w = s_1\ \cdots \ s_o\ s_e$ for words $s_i\in\Sigma^*$, $i\in[o]$, representing the states, and a technical final word $s_e\in\Sigma^*$ that is also responsible for transcription termination. From now on, for each newly introduced word, assume that its specification or characterization via constraints is given later on in the proof (if it is not given immediately or before). Also, assume that any word that occurs as one side of a context in the context-set $C$ (which is to be constructed later) does not appear anywhere else by other overlapping factors in $w$.
    
    We start with the construction of $s_e$. Let $f_s,f_e\in\Sigma^*$ be two words such that $|f_sf_e| = 2m$ and $f_s = \theta(f_e)$, hence, $f_sf_e$ forms a stem of length $2m$, i.e., $f_e = \theta(f_s)$. Let $T = \{t\}$ be a set of terminating sequences containing only one word $t\in\Sigma^*$. We set the suffix $s_e$ by
    $$s_e = \theta(S_{end})E_{end}\ f_e\ t\ \theta(S_{q_1})E_{q_1}$$
    for words $S_{end},E_{end},E_{q_1}\in\Sigma^*$. Next, for each state $q_i\in Q$, we construct a word $s_i\in\Sigma^*$ by setting
    $$ s_i = S_{i,1}\ \cdots \ S_{i,k}\ S_{i,k+1}\ e_{i,1}\ \cdots \ e_{i,k}\ e_{i,k+1}\ \theta(S_{q_{i+1}})\ E_{q_{i+1}} $$
    for words $S_{q_{i+1}},E_{q_{i+1}},S_{i,1},\dots ,S_{i,k+1},e_1,\dots ,e_{k_+1}\in\Sigma^*$, with $k$ being the number of outgoing transitions of $q_i$. If $q_i = q_o$, then $S_{q_{i+1}} = E_{q_{i+1}} = \varepsilon$ is just an empty suffix. Assume some arbitrary ordering on the outgoing transitions of $q_i$ and the transition labels to be given by $\ta_{i,1}$ to $\ta_{i,k}$.
    Then, each $e_{i,j}$, for $j\in[k]$, is defined by
    $$ e_{i,j} = \theta(S_{i,1}\cdots S_{i,j})\ E_{i,j}\ \ta_{i,j}\ S_{q_{i,k}} $$
    for the words $S_{q_{i,j}}E_{i,j}\in\Sigma^*$. Suppose that $S_{i,j} = S_{q_{j'}}$ for some $j'\in[o]$, let $q_{i,j}$ be the state reached by the $\nth{j}$ transition of $q_i$. If we have $q_i\in Q\setminus F$, i.e., it is not a final state, then we set $S_{i,k+1} = e_{i,k+1} = \varepsilon$ to be empty. If we have $q_i\in F$, i.e., it is a final state, then $S_{i,k+1}\in\Sigma^+$ is a nonempty word and $e_{i,k+1}$ is given by
    $$e_{i,k+1} = \theta(S_{i,1}\cdots S_{i,k+1})\ E_{i,k+1}\ f_s\ S_{end}$$
    for the word $E_{i,k+1}\in\Sigma^*$ and the others as defined above. Notice that the only differences to the other $e_j$'s are, first, the fact that, instead of a letter $\ta$, we write the word $f_s$ and, second, instead of the suffix $S_{q_{i,j}}$, we add the word $S_{end}$. This concludes all structural elements.
    
    We continue with the definition of the context-set $C$ over the words defined above. For each $i\in [o]$, we assume $(S_{q_i},E_{q_i})\in C$. These are contexts responsible for jumping between different $s_i$ and $s_j$ using hairpin deletion. In addition to $i\in[o]$, for each $j\in[k]$ (or $j\in[k+1]$ if $q_i\in F$), $k$ being the number of outgoing transitions of the state $q_i$ as before, we assume $(S_{i,1}\cdots S_{i,j},E_{i,j})\in C$. Finally, we assume $(S_{end},E_{end})\in C$, concluding the definition of $C$.
    
    Again, we mention that we assume that each occurrence of a left or right context in $w$ is exactly given by the above construction. More occurrences resulting from overlapping factors are excluded by assumption. Also, for each context $(\ell,r)\in C$, we assume that each left context $\ell$ is chosen long enough so that it forms a valid hairpin with $\theta(\ell)$ regarding the logarithmic energy model for each occurrence of $\ell$ with the closest occurrence of the right context $r$ after it. Keep in mind that each left context has a unique corresponding right context, so no non-deterministic choice can happen here. We continue with an intuitive explanation of the functionality of the construction.

    \textbf{Intuition:} Transcription starts in $s_1$ and moves around $w^\omega$ over and over again (see Figure~\ref{fig:simulation-circular-transcription}). Every time the transcription is at the beginning of one of the factors $s_i$, $i\in[o]$, a nondeterministic choice of one of the encoded transitions occurs. By the selection of one of the factors $S_{i,1}$, $S_{i,1}S_{i,2}$, \dots  , $S_{i,1}S_{i,2}\cdots S_{i,k}$ as a left context, we are forced to use the corresponding right context $E_{i,j}$ (if the prefix $S_{i,1}\cdots S_{i,j}$ is chosen, $j\in{k}$). Everything in between gets removed by bounded hairpin deletion, using the stem $S_{i,1}\cdots S_{i,j}$ with corresponding $\theta(S_{i,1}\cdots S_{i,j})$ and a loop containing everything in between. This represents the choice of an outgoing transition of $q_i$. Immediately after the removed hairpin, the letter $\ta_{i,j}$ occurs in $s_i$. See Figure~\ref{fig:simulation-transition-selection} for a visualization of this process.
    
    \begin{figure}[h]
    \centering
    \includegraphics[width=14cm]{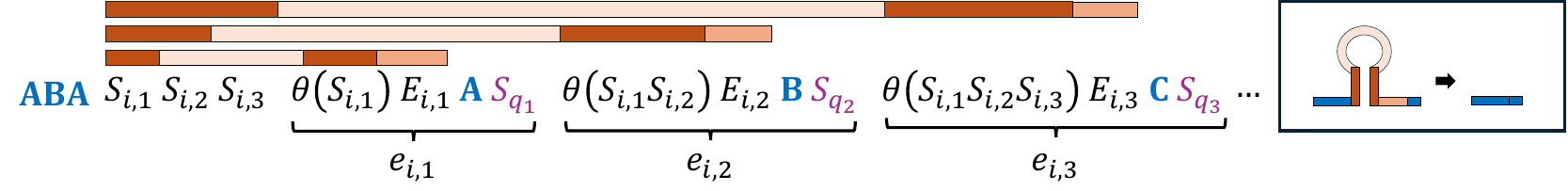}
    \caption{Nondeterministic selection of a transition out of the encoded $q_i\in Q$ in $s_i$. This example assumes $3$ outgoing transitions and the respective letters $\mathtt{A}$, $\mathtt{B}$, and $\mathtt{C}$. Depending on the choice of the left context (dark color - left), a corresponding right context (medium color - right) can be chosen and a hairpin with the corresponding $\theta$ part (dark color - right) is formed, having everything in between as part of the loop (bright color - middle). Finally, the whole marked region is then removed by hairpin deletion, leaving the blue marked letter as a new prefix.}
    \label{fig:simulation-transition-selection}
    \end{figure}
    
    Keep in mind, that the lengths of $S_{i,j'}$, $j'\in[j]$, and $E_{i,j}$ have to be adapted accordingly for this to work. Due to maximally parallel hairpin deletion, the next left context in $w$, in particular in $s_i$, has to be chosen for hairpin deletion. The immediate next left context is $S_{\delta(q_i,\ta_j)}$ from which we have to choose the unique right context $E_{\delta(q_i,\ta)}$ to jump to $s_{q_{i,k}}$. This is the aforementioned jump between $s_i$ and $s_{q_{i,k}}$. See Figure~\ref{fig:simulation-state-jump} for a visualization of this process.

    \begin{figure}[h]
    \centering
    \includegraphics[width=11cm]{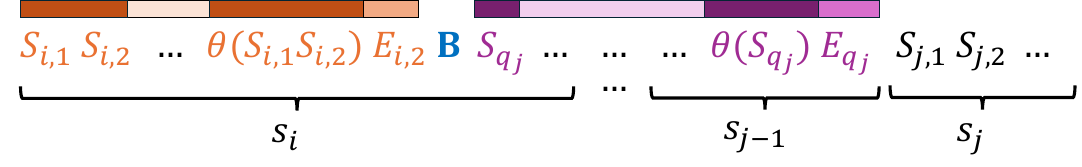}
    \caption{Representation of a jump between $2$ state encodings $s_i$ and $s_j$. Assuming some transition to be selected and the corresponding part to be removed by hairpin deletion (orange marked region), a deterministic left and right context selection results in everything between the letter $\mathtt{B}$ and the state encoding $s_j$ to re removed (purple marked region).}
    \label{fig:simulation-state-jump}
    \end{figure}
    
    After these two steps (transition selection and jump to the next state), we can repeat this process. A prefix of a word in the language $L(A)$ is successively obtained by maximally parallel bounded hairpin deletion. To terminate transcription, in a final state, we can use the same mechanism to choose the left context $S_{i,1}S_{i,2}\cdots S_{i,k}S_{i,k+1}$ and the right context $E_{i,k+1}$ to obtain a suffix stem $f_sf_e$ which is followed by $t\in T$ in $w^\omega$, which results in transcription termination. See Figure~\ref{fig:simulation-final-state} for a visualization.
    
    \begin{figure}[h]
    \centering
    \includegraphics[width=12cm]{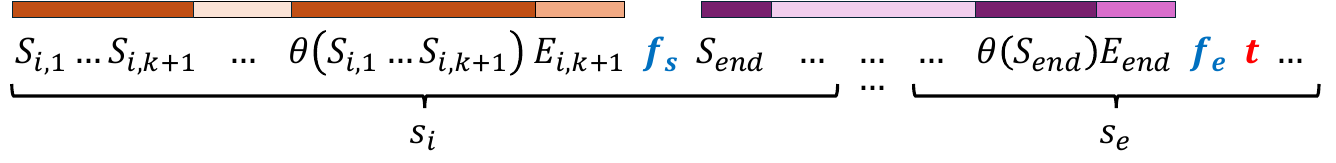}
    \caption{Assuming $s_i$ encodes a final state with $k$ outgoing transitions, then a transition $k+1$ can be taken to obtain $f_s$ as a suffix after some previously obtained word (orange marked region). From there, a deterministic choice of left and right contexts $S_{end}$ and $E_{end}$ allows for a jump to $s_e$ (purple marked region), resulting in the suffix $f_sf_et$. As $t\in T$ and $f_sf_e$ forms a stem with length $2m$ by assumption, transcription can be terminated here.}
    \label{fig:simulation-final-state}
    \end{figure}
    
    Hence, as termination can only be initiated from final $s_i$'s that represent final states $q_i\in F$, all words obtainable from $w^\omega$ by maximally parallel bounded hairpin deletion must be words in $L(A)$. As all transitions are encoded and available to be chosen in the beginning of each $s_i$, we can also obtain all words in $L(A)$, concluding this sketch. As mentioned before, due to space constraints, a full formal proof of correctness, i.e., that  $L(A) = [w^\omega]_{\hpdblogpmax_{S'}}$, can be found in Appendix~\ref{section:appendix-simulation-proof}. An engineered example that uses actual words over the RNA alphabet $\Sigma_{\mathtt{RNA}} = \{\tA,\tU,\tC,\tG\}$ can be found in the  Appendix~\ref{sec:appendix-simulation-example-specific}.
\end{proof}

For this construction to work, we clearly needed some assumptions regarding the hairpin deletion model to simulate co-transcriptional splicing. These are especially the assumptions that we are able to do this in a parallel manner (modeling the independence of subsequent co-transcriptional deletions) and that we have some level of greediness in the system (maximally parallel bounded hairpin deletion). Without some greediness assumption regarding the left context, an infinite number of points where deletion could happen can be skipped, and we could start simulating a computation on $A$ at the beginning of every iteration of $w^\omega$. We continue with a general framework that implements the construction from Theorem~\ref{theorem:dfa-simulation} using the RNA alphabet $\Sigma_{\mathtt{RNA}} = \{\tA,\tU,\tC\tG\}$, sketching the existence of working encodings.

\subsection{Applying Theorem~\ref{theorem:dfa-simulation} for the RNA-Alphabet $\Sigma_{\mathtt{RNA}}=\{A,U,C,G\}$}
\label{sec_sim_rna-framework}

For any working encoding, the following four questions need to be answered:
\begin{itemize}
\item[1.] How to encode the internal transition selection in each $s_i$?
\item[2.] How to realize the transitions between states, i.e., the jump between arbitrary $s_i$ and $s_j$?
\item[3.] How to realize termination of the parallel hairpin deletion process?
\item[4.] Do the words that encode contexts $(x,y)\in C$ appear only in the intended positions?
\end{itemize}
In the main part, we discuss these questions in a more general manner. A specific example of encoding a specific NFA using this general framework can be found in Appendix~\ref{section:appendix-simulation-example}.

Consider some NFA $A = (Q,\Sigma_{\mathtt{RNA}},q_1,\delta,F)$ with $Q = \{q_1,\dots ,q_o\}$, for some $o\in\N$. Let $q_i\in Q$, for $i\in[o]$, be some state in $Q$. Assume $q_i$ has $k$ transitions, for some $k\in\N$. For each $j\in[k]$, let $\ta_{i,j}$ be the letter on the $\nth{j}$ transition of $q_i$ and let $q_{i,j}\in\delta(q_i,\ta_{i,j})$ be the resulting state from $q_i$ by taking the $\nth{j}$ transition.

(1) First, we define $\theta$ to represent the strongest base pair bonds, i.e., $\theta(\tA) = \tU$, $\theta(\tU) = \tA$, $\theta(\tC) = \tG$, and $\theta(\tG) = \tC$. To encode $q_i$ and its outgoing transitions in a word $s_i\in\Sigma_{\mathtt{RNA}}^*$, generally, $s_i$ will have the form as described in the construction given in Theorem~\ref{theorem:dfa-simulation}, i.e., if w.l.o.g. $q_i\notin F$,
$$ s_i = S_{i,1}\cdots S_{i,k}\ e_{i_1}\cdots e_{i_k}\ \theta(S_{q_{i+1}})E_{q_{i+1}}. $$
We can set 
$$S_{i,1}\cdots S_{i,k} = \tA\tA(\tC^{i})\tA\tA(\tG\tA\tA)^k.$$
Specifically, we set $S_{i,1} = \tA\tA(\tC)^i\tA\tA\tG\tA\tA$ and, for each $j\in[k]$ with $j > 1$, we set $S_{i,k} = \tG\tA\tA$. For example, if $q_i$ has $3$ transitions, then $$S_{i,1}\cdots S_{i,3} = \tA\tA(\tC)^i\tA\tA\tG\tA\tA\tG\tA\tA\tG\tA\tA,$$
and we have $S_{i,1} = \tA\tA(\tC)^i\tA\tA\tG\tA\tA$ and $S_{i,2} = S_{i,3} = \tG\tA\tA$. Now, by the definition of $\theta$, we obtain, for each $j\in[k]$, 
$$\theta(S_{i,1}\cdots S_{i,j}) = (\tU\tU\tC)^j\tU\tU(\tC)^i\tU\tU.$$
From before, we know that each $e_{i,j}$, $j\in[k]$, we have
$$ e_{i,j} = \theta(S_{i,}\cdots S_{i,j})\ E_{i,j}\ \ta_{i,j}\ S_{q_{i,j}}.$$
We obtain $\theta(S_{i,}\cdots S_{i,j}$) implicitly by the previous definition. For simplicity reasons, we can set $E_{i,j} = \theta(S_{i,}\cdots S_{i,j})$ (hence, equal to the $\theta(S_{i,}\cdots S_{i,j})$ part). This concludes already the implementation of the contexts $(S_{i,1}\cdots S_{i,j},E_{i,j})\in C$. The letter $\ta_{i,j}$ is just the letter of the transition. The final parts of $s_i$ that need to be defined are $S_{q_{i,j}}$ and $E_{q_{i+1}}$. More generally, we need to define the words for each context $(S_{q_m},E_{q_m})\in C$, for $m\in[o]$. Let $o_m\in\N$ be a positive number for each state $q_m\in Q$. Then, we set $S_{q_m} = \tU\tU\tU(\tG)^{o_m}\tU\tU\tU$ and $E_{q_m} = \theta(S_{q_m})$. We discuss the numbers $o_m$ later when discussing the second question from above. With that, we have defined all parts of $s_i$.

To select any transition, pick a prefix $S_{i,1}\cdots S_{i,j}$, $j\in[k]$, find the corresponding right context $E_{i,j}$, and then notice that before $E_{i,j} = \theta(S_{i,1}\cdots S_{i,j})$, there is another occurrence of $\theta(S_{i,1}\cdots S_{i,j})$. We obtain the following hairpin structure (red) with corresponding left and right contexts and the remaining letter $\ta_{i,j}$:
$$
\underbrace{{\color{red} \tA\tA(\tC)^i\tA\tA(\tG\tA\tA)^j}}_{S_{i,1}\dots S_{i,j}}\ldots\underbrace{{\color{red}(\tU\tU\tC)^j\tU\tU(\tG)^i\tU\tU}}_{\theta(S_{i,1}\dots S_{i,j})}\underbrace{(\tU\tU\tC)^j\tU\tU(\tG)^i\tU\tU}_{E_{i,j}}\ {\color{blue} \ta_{i,j}\ S_{q_{i,j}}}
$$
If, due to the logarithmic energy model, the length of the inner part between $S_{i,1}\dots S_{i,j}$ and $\theta(S_{i,1}\dots S_{i,j})$ gets too large, in the position of the last $\tG$ in $S_{i,1}\dots S_{i,j}$ (and in the position of the first $\tC$ in $\theta(S_{i,1}\dots S_{i,j})$, we can increase the number of $\tG$'s (and $\tC$'s respectively) until the stem is supports the length of the loop. We might have to do this for multiple transitions in a single state encoding $s_i$ until the values balance out, but due to the exponential increase in supported size per letter in the stem, we can always reach a length that works out. This concludes the question on how to encode the transitions selection using the alphabet $\Sigma_{\mathtt{RNA}}$.

(2) Next, we briefly discuss, how to realize the jump between transition-encodings $s_i$ and $s_m$. Notice, that we have already defined the encoding of the related contexts $(S_{q_m},E_{q_m})\in C$, $m\in[o]$, by $S_{q_m} = \tA\tA\tA(\tC)^{o_m}\tA\tA\tA$ and $E_{q_m} = \theta(\tA\tA\tA(\tC)^{o_m}\tA\tA\tA) = \tU\tU\tU(\tG)^{o_m}\tU\tU\tU$. As before with the number of $G$'s and $C's$ in the transition encodings, we left the number of $G$'s and $C$'s in this case open as well. Depending on the number of letters between each left and right context $(S_{q_m},E_{q_m})\in C$, a different number might be needed to obtain a stem that supports the length of the loop. In the logarithmic energy model, we can always find such numbers as a linear increase in stem size results in an exponential increase of supported loop size. In the example case above, notice that after ${\color{blue}\ta_{i,j}}$, there is a unique occurrence of $S_{q_{i,j}} = \tA\tA\tA(\tC)^{o_m}\tA\tA\tA$, assuming $q_m = q_{i,j}$. The unique occurrence of $E_{q_m}$ is right in front of $s_m$. By the construction, we obtain the following unique hairpin formation. Here, we have no overlapping left contexts. So, no nondeterministic choice is possible, This hairpin deletion step must occur, if the letter before is used in an obtained word:
$$
{\color{blue} \ta_{i,j}}\ \underbrace{{\color{red} \tA\tA\tA(\tC)^{o_m}\tA\tA\tA}}_{S_{q_m}}\ldots\underbrace{{\color{red}\tU\tU\tU(\tG)^{o_m}\tU\tU\tU}}_{\theta(S_{q_m})}\underbrace{\tU\tU\tU(\tG)^{o_m}\tU\tU\tU}_{E_{q_m}}\ {\color{blue} s_m}
$$
This concludes the discussion on how to realize the jumps between state encodings $s_i$ and~$s_m$. 

(3) Finally, we need to encode the part that results in termination of the hairpin deletion process. For that, we define $s_e$, i.e., the words $f_s$, $f_e$, $t$, and in particular the words in the context $(S_{end},E_{end})\in C$. In principle, we can consider $s_e$ as its own state encoding but without any transitions. So, similar to the contexts $(S_{q_m},S_{q_m})\in C$, for $m\in[o]$, we can define, for some $o_e\in\N$ that is not equal to any other $o_m$, $S_{end} = \tA\tA\tA(\tC)^{o_e}\tA\tA\tA$ as well as $E_{end} = \theta(S_{end}) = \tU\tU\tU(\tG)^{o_e}\tU\tU\tU$. Additionally, we set $t = (\tU)^{10}$, $f_s = (\tA)^4\tG(\tA)^4$, and $f_e = (\tU)^4\tC(\tU)^4 = \theta(f_s)$. In the encoding of each final state $s_i$, for some $q_i\in F$, we handle $S_{i,k+1}$, assuming $q_i$ has $k$ outgoing transitions, and $E_{i,k+1}$ analogously to any other transition and add $f_s$ in the place where the letter of a transition would have been. This results in the following 2 hairpin deletion steps in a final state. Assuming $w\in\Sigma^*$ is a word in $L(A)$ which has been obtained already as a prefix and $q_i$ being the final state reached after reading $w$, the following becomes a general possibility:
$$
{\color{blue} w}\ 
\underbrace{{\color{red} \tA\tA(\tC)^{i}\tA\tA(\tG\tA\tA)^{k+1}}}_{S_{i,1}\cdots S_{i,k+1}}
\ldots
\underbrace{{\color{red} (\tU\tU\tC)^{k+1}\tU\tU\tG\tU\tU}}_{\theta(S_{i,1}\cdots S_{i,k+1})}\ 
\underbrace{(\tU\tU\tC)^{k+1}\tU\tU\tG\tU\tU}_{E_{i,n+1}}\ 
\underbrace{{\color{blue}(\tA)^4\tG(\tA)^4}}_{f_s}\ 
\underbrace{{\color{blue}\tA\tA\tA(\tG)^{o_e}\tA\tA\tA}}_{S_{end}}
$$
$$
{\color{blue} w}\ 
\underbrace{{\color{blue}(\tA)^4\tG(\tA)^4}}_{f_s}\ 
\underbrace{{\color{red}\tA\tA\tA(\tG)^{o_e}\tA\tA\tA}}_{S_{end}}\ 
\ldots\ 
\underbrace{{\color{red}\tU\tU\tU(\tC)^{o_e}\tU\tU\tU}}_{\theta(S_{end})}\ 
\underbrace{{\color{black}\tU\tU\tU(\tC)^{o_e}\tU\tU\tU}}_{E_{end}}\ 
\underbrace{{\color{blue}(\tU)^4\tC(\tU)^4}}_{f_e}\ 
\underbrace{{\color{orange}\tU\tU\tU\tU\tU\tU\tU\tU\tU\tU}}_{t}
$$

(4) We can see that, if only the intended hairpin deletion steps are possible then this encoding effectively works, assuming the words in the contexts are pumped to be big enough. However, it needs to be made sure that hairpin deletion and context recognition can only occur in the intended positions and that no other factor can be used for that purpose. By checking the encodings, one can make sure, that no left or right context appears in unintended positions as a factor. See Appendix~\ref{sec:appendix-simulation-example-general-examination} for a detailed examination on why this is the case. 

This concludes all elements needed to encode an arbitrary NFA $A$ in some circular word $w^\omega$ such that $L(A) = [w]_{\hpdblogpmax}$. As mentioned before, in Appendix~\ref{sec:appendix-simulation-example-specific}, an exemplary encoding of a specific NFA using this framework is given. 

This concludes this section. It has been shown that the mechanism of log-hairpin deletion, modeling co-transcriptional splicing, can be used to encode any (infinite) regular language by encoding its DFA or NFA representation on a finite circular DNA word. As the construction grows with the size of the DFA or NFA at hand, minimizing the input NFAa is a second problem that can be looked at. Due to the practical nature of this problem, even restricted models of NFAs could be considered.

\section{Minimizing Input NFA's}
\label{sec_dec}
The construction from the previous section makes it possible to produce any given finite language from powers of a word~$w$ by parallel hairpin deletion. As our goal is to engineer DNA templates that efficiently encode a given set of sequences to be produced by co-transcriptional splicing, Theorem~\ref{theorem:dfa-simulation} shows that we can achieve our goal, and we can focus on optimizing the length of $w$, i.e., the period. To allow further efficiency gains and some leeway for possible lab implementations of the construction, we can relax the requirement that a given finite set is produced \emph{exactly}, that is, without any other words obtained by the system. With implementation in mind, we require the `extra' sequences produced to not interfere with the words in the target set of sequences, i.e., that they do not share hybridization sites with our target. 

We can formalize this as a set of forbidden patterns to be avoided as factors by the `extra' words. In some cases, we may also assume that sequences above given lengths can be efficiently filtered out from the resulting set. We propose the following problem(s).

\begin{problem}[Small automaton for cover set avoiding forbidden patterns]\label{problem:small-cover-set-avoiding-forbidden-patterns}
    Given a set $W=\{w_1,\dots,w_k\}$ of target words and a set $F=\{f_1,\dots,f_\ell\}$ of forbidden patterns, such that each $f_i\in F$ is a factor of some $w_j\in W$, find a smallest NFA $M$, such that:
    \begin{itemize}
        \item $W\subseteq L$, and
        \item $\left( L\setminus W \right) \cap \Sigma^* F\Sigma^* = \emptyset$,
    \end{itemize}
    where either
    \begin{itemize}[leftmargin=2cm]
        \item[(exact)] $L=L(M)$, in the general case, or
        \item[(cover)] $L=L(M)\cap \Sigma^{\leq n}$ for $n=\max \{|w| \mid w\in W\}$, in the length bounded setting of the problem.
    \end{itemize} 
\end{problem}

Problem \ref{problem:small-cover-set-avoiding-forbidden-patterns} can be either considered as a minimization problem for the number of states or a minimization problem for the number of transitions (notice that the size of the encoding in Section \ref{sec_sim} primarily depends on the number of transitions). The following example shows that the problem is not trivial in the sense that there exist target and forbidden pattern sets $W,F$ for which the smallest NFA $M$ satisfying the conditions is smaller than both the minimal NFA for $W$ and the minimal NFA for $\overline{\Sigma^*F\Sigma^*}\cup W$. 

\begin{example}
    Let $W=\{aab^kaa\}$ and $F=\{ab^ka\}$. It is easy to see that the minimal NFA for $W$ must have at least $k+5$ states and $k+4$ transitions, otherwise it would have a loop which reads $b^k$ and it would accept an infinite language. Similarly, the minimal NFA for the language containing $W$ but avoiding all other words having forbidden patterns from $F$,  that is, $\overline{\Sigma^*F\Sigma^*}\cup W$, has at least $|Q| \in \Omega(k)$ states and $\Omega(k)$ transitions, otherwise, for accepting inputs of the form $ab^jaa$, with $|Q|\leq j< k$ the block of $b$'s would be read by some cycle labeled $b^\ell$. Together with the fact that $ab^{k-\ell}aa$ must be accepted, we get that the machine would also accept $ab^kaa$, a forbidden word. However, a simple automaton $M$ accepting $aab^*aa$ with $5$ states and $5$ transitions satisfies the conditions, $W\subset L(M)$ and $L(M)$ does not include any other word with a factor $ab^ka$.
\end{example}

As the number of states implicitly gives an upper bound on the possible number of transitions of an NFA, the hardness of Problem~\ref{problem:small-cover-set-avoiding-forbidden-patterns}(cover, states), which is the version of the problem where we are looking for the smallest cover automaton in terms of the number of its states rather than transitions, suggests that all versions of Problem~\ref{problem:small-cover-set-avoiding-forbidden-patterns} are NP-hard, as similar straightforward reductions are possible with simple target and forbidden sets. All proofs in this section can be found in Appendix~\ref{section:appendix-proofs-ssbnfas}.

\begin{proposition}\label{prop:hardness-small-cover-set-avoiding-forbidden-patterns-problem}
    Problem \ref{problem:small-cover-set-avoiding-forbidden-patterns}(cover, states) is NP-hard.
\end{proposition}

\subsection{Practically Motivated Restricted NFA-models}

As it is practically impossible to obtain RNA sequences of infinite size from circular DNA transcription, we propose the following NFA-like models which impose certain restrictions on valid computations. Encoding finite languages using those variants can reduce the size of the resulting machine w.r.t. to the minimal classical NFA encoding. Initially we also hoped that those restricted models would allow for more tractable minimization algorithms. In what follows, however, we show the NP-hardness for the minimization problem for all models we considered. The first restriction comes from the fact that the number of times the template word is `read' in circular transcription might be limited. There are various ways of expressing this limitation in terms of the computation of the automata. First, one can restrict the number of times that a certain state of the automaton can be reached in a computation.

\begin{definition}
    A \emph{state-step-bounded nondeterministic finite automaton} (SSB-NFA) $A$ is a hexatuple $A = (Q,c,\Sigma,\delta,q_0,F)$
    where $Q$ is a \emph{finite set of states}, $c\in\N$ denotes the \emph{state-step-bound},
    $\Sigma$ denotes a \emph{finite alphabet}, $\delta : Q\times\Sigma \rightarrow Q$ denotes the transition function,
    $q_0\in Q$ the initial state, and $F\subseteq Q$ the set of final states.

    The language $L(A)$ of a SSB-NFA is the language of all words $w\in\Sigma^*$ where $\delta(q_0,w)\cap F\neq \emptyset$
    and $w$ has some accepting computation where each state $q_i$ occurs at most $c$ times.
\end{definition}

Next, in the encoding of NFAs onto templates for hairpin deletion proposed in Theorem~\ref{theorem:dfa-simulation}, different states are encoded consecutively. So, simulating a transition from an earlier encoded state to a later encoded state still occurs on a single repetition of the template. Hence, we could also impose an order $<_Q$ on the states and restrict the number of times that $q_i$ follows $q_j$ in a computation if $q_i <_Q q_j$, effectively resulting in having to use another repetition in $w^\omega$. This results in the notion of \emph{return-bounded nondeterministic finite automata} (RB-NFAs). Their formal definition is given in Appendix~\ref{section:appendix-proofs-ssbnfas}. Finally, another practical limitation might impose a bound on the length of the intron that is removed during hairpin deletion. This can be represented by imposing an order $<_Q$ on $Q$, setting a distance for each $2$ subsequent elements over $<_Q$, and setting a bound on the maximum forward distance between two subsequent states in a computation over the automaton. This results in the notion of \emph{distance-bounded nondeterministic finite automata} (DB-NFAs). Their formal definition is given in Appendix~\ref{section:appendix-proofs-ssbnfas} as well. The class of all SSB-NFAs, RB-NFAs, and DB-NFAs is denoted by $C_{FA}$. We will investigate SSB-NFAs as an exemplary case and see that we can obtain NP-hardness results for the decision variant of the minimization problem for all models in the class $C_{FA}$ using very similar proofs.

\subsection{Properties of State-Step-Bounded NFAs}

SSB-NFAs restrict the number of times we are allowed to be in a specific state. Hence, for SSB-NFAs (and RB-NFAs) $A$, we know that $L(A)$ is a finite language. One big advantage of SSB-NFAs is their potential to be significantly smaller than NFAs recognizing the same language. When encoding a finite language as NFA, each word in the language must occur as the label of a simple path in the state graph as any loop in the NFA with a path to a final state results in an infinite language. In SSB-NFAs, certain repetitions may be represented with a small loop. Consider the following example. 

\begin{example}\label{example:ssb-nfa-space-savings}
    Let $L = \{\ta,\ta\tb\ta,\ta\tb\ta\tb\ta,\ta\tb\ta\tb\ta\tb\ta,\ta\tb\ta\tb\ta\tb\ta\tb\ta\}$. Let $A = (Q,\Sigma,\delta,q_0,F)$ be
    a minimal NFA represented by the left state diagram in Figure \ref{fig:ssb-nfa-space-saving} such that $L = L(A)$. We see that there exists a (minimal) SSB-NFA $B = (Q',5,\Sigma,\delta',q_0',F')$ represented by the state diagram on the right in Figure \ref{fig:ssb-nfa-space-saving} for which we also have $L(B) = L = L(A)$, but its size is significantly smaller than $A$ regarding the number of transitions and states.
    \begin{figure}[h]
    \centering
    \includegraphics[width=11cm]{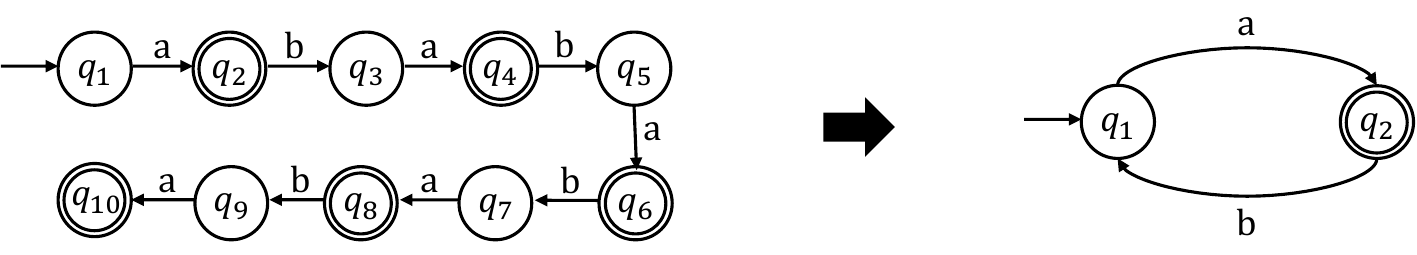} 
    \caption{One minimal NFA $A = (Q,\Sigma,\delta,q_0,F)$ (left) and  minimal SSB-NFA $B = (Q',5,\Sigma,\delta',q_0',F')$ (right) recognizing the same language $L$ found in example \ref{example:ssb-nfa-space-savings}.}
    \label{fig:ssb-nfa-space-saving}
    \end{figure}
\end{example}

Indeed, we can find for any NFA $A$ accepting a finite language a SBB-NFA $B$ which is smaller or equal to the size of $A$ in terms of the number of states or transitions but which accepts the same language. Trivially, as a NFA recognizing a finite language has no loops on paths that reach final states, we can essentially copy the whole automaton and set the step-bound $c$ to any value to obtain a SSB-NFA which recognizes the same langauge. As we are interested in finding minimal templates for contextual splicing, we propose the following decision variant of the SSB-NFA minimisation problem.

\begin{problem}[SSB-NFA-Min]\label{problem:ssb-nfa-min}
    Let $L$ be some regular language and let $c,k\in N$ be some positive integers.
    Does there exist some SSB-NFA $A = (Q,c,\Sigma,\delta,q_0,F)$ such that $L(A) = L$ with either $|Q| \leq k$ (SSB-NFA-Min-States)
    or $|\delta| \leq k$ (SSB-NFA-Min-Transitions).
\end{problem}

We know that the answer to both problems is $\mathtt{false}$ if the input language $L$ is an infinite language. This can be efficiently checked for any representation, i.e., DFAs, NFAs, or regular expressions. Due to the sizes of different representations of regular languages, we obtain the following results.

\begin{proposition}\label{lemma:ssb-nfa-min-in-np}
    $\mathtt{SSB\text{-}NFA\text{-}Min}$ is in NP if the input language $L$ is given as a finite list of words. If $L$ is presented as a regular expression or NFA, then the problem is in PSPACE.
\end{proposition}

Similar to the decidability version of the problem of finding minimal NFAs for finite languages (see \cite{DBLP:conf/lata/GruberH07,Amilhastre2001}), we can show that $\mathtt{SSB\text{-}NFA\text{-}Min}$ is NP-hard by a reduction from the biclique covering problem.

\begin{proposition}\label{lemma:ssb-nfa-min-hard}
    $\mathtt{SSB\text{-}NFA\text{-}Min}$ (states/transitions) is NP-hard for all alphabets $\Sigma$ with $|\Sigma| \geq 2$.
\end{proposition}

By that, we can conclude the results of this section in the following theorem.

\begin{theorem}
    The $\mathtt{SSB\text{-}NFA\text{-}Min}$ problem is NP-complete for all alphabets $|\Sigma| \geq 2$ if the input language is given as a list of words.
\end{theorem}

For the minimization problems of RB-NFAs and DB-NFAs, almost the exact same reduction as for SSB-NFAs can also be applied to obtain NP-hardness in these models. 
A sketch of the differences is also given in Appendix~\ref{section:appendix-proofs-ssbnfas}.
For the following result, assume that $\mathtt{RB\text{-}NFA\text{-}Min}$ and $\mathtt{DB\text{-}NFA\text{-}Min}$ are defined analogously to $\mathtt{SSB\text{-}NFA\text{-}Min}$.

\begin{proposition}\label{proposition:rb-nfa-min-np-hard}
    $\mathtt{RB\text{-}NFA\text{-}Min}$ and $\mathtt{DB\text{-}NFA\text{-}Min}$ are NP-hard for all alphabets $\Sigma$ with $|\Sigma| \geq 2$.
\end{proposition}

This concludes the results regarding the automata models in $C_{FA}$.
As even in these restricted cases we still have NP-hardness for minimization, this motivates identifying target languages with restricted structure for which efficient minimization algorithms may exist.

\section{Conclusion}
In this paper, we provided a framework which shows that any regular language can be obtained from circular words using maximally bounded parallel logarithmic hairpin deletion. 
This indicates that co-transcriptional splicing may be utilized to encode and obtain any set of RNA sequences representable by regular languages, even potentially infinite ones, from a circular DNA template of finite size. As the construction is based on the NFA representation of the encoded regular language, the problem of minimizing NFAs has been further investigated, in particular for well motivated restricted NFA-like models (SSB-NFAs, RB-NFAs, and DB-NFAs). 
As only hardness results could be obtained so far, future research could work on identifying practically motivated classes of languages for which we can efficiently find small NFAs, SSB-NFAs, RB-NFAs, or DB-NFAs. In addition to that, another future challenge lies in applying hairpin deletion to obtain language classes with higher expressibility, e.g. context-free or context-sensitive languages. But whether this is possible is left as an open question for which we have no specific conjecture so far.

\begin{question}
    Let $L_G$ be some context-free language. Does there exist a word $w\in\Sigma^*$ and a properties tuple for log-hairpin deletion $S$, such that $L(A) = [w^\omega]_{\hpdblogpmax_S}$?
\end{question}

\bibliographystyle{plain}
\bibliography{programmable_cotranscriptional_splicing}

\newpage
\appendix

\section{Proof of Correctness for the Construction in Theorem~\ref{theorem:dfa-simulation}}
\label{section:appendix-simulation-proof}
In this section, a formal correctness proof for the construction given regarding Theorem~\ref{theorem:dfa-simulation} is provided. First, we show that $L(A)\subseteq[w^\omega]_{\hpdblogpmax_{S'}}$ given a constructive argument. Then, we show that $[w^\omega]_{\hpdblogpmax_{S'}}\subseteq L(A)$ by an inductive argument of the hairpin deletion steps that have to be taken in order for parallel hairpin deletion to terminate.

\textbf{Formal proof of first direction ($L(A)\subseteq[w^\omega]_{\hpdblogpmax_{S'}}$):} Assume $u\in L(A)$. Assume $q_{i_1}q_{i_2}...q_{i_{|u|+1}}$ to be a sequence of states to obtain $u$ in $L(A)$, i.e., such that $q_{i_1} = q_1$, $q_{|u|+1}\in F$, and for all $i_j,i_{j+1}$, $j\in[|u|]$, we have $q_{i_{j+1}}\in\delta(q_{i_j},u[j])$. Assume $k_j$ to mark the number of the transition taken to obtain $q_{i_{j+1}}$ from $q_{i_j}$ with $u[j]$. By the construction above, we know that we can factorize $w^\omega$ into
$$ v_1\ u[1]\ v_{2,\ell}\ v_{2,r}\ u[2]\ ...\ u[|u|-1]\ v_{|u|,\ell}\ v_{|u|,r}\ u[|u|]\ v_{|u|+1,\ell}v_{|u|+1,r}\ f_s\ v_{|u|+2}\ f_e\ t\ yw^\omega $$ such that $t\in T$, $y = \theta(S_{q_1})E_{q_1}$, and
\begin{align*}
    v_1             &= S_{1,1}...S_{1,\ell_1}\ ...\ \theta(S_{1,1}...S_{1,k_1})E_{1,\ell_1}\text{, } \\
    v_{j,\ell}      &= S_{q_{i_{j+1}}}...\theta(S_{q_{i_{j+1}}})E_{q_{i_{j+1}}}\text{, for $j\in[|u|+1]\setminus\{1\}$, } \\
    v_{j,r}         &= S_{j+1,1}...S_{j+1,k_{j+1}}\ ...\ \theta(S_{j+1,1}...S_{j+1,k_{j+1}})E_{k_{j+1}}\text{, for $j\in[|u|]\setminus\{1\}$, } \\
    v_{|u|+1,r}     &= S_{i_{|u|+1},1}...S_{i_{|u|+1},k_{|u|+1}}\ ...\ \theta{S_{i_{|u|+1},1}...S_{i_{|u|+1},k_{|u|+1}}}E_{|u+1|}\text{, and} \\
    v_{|u|+2}       &= S_{end}\ ...\ \theta(S_{end})E_{end}.
\end{align*}
Each word $v$ marks a single factor completely removed by a single hairpin-removal step. As each right context appears uniquely in $w$, we assume that the first following occurrence in $w^\omega$ is taken. As we always pick a left context right after a removed factor and as each $u[i]$ cannot be part of a left context by assumption in the construction, we know that this factorization is valid for maximal parallel bounded hairpin deletion. Intuitively, $v_1$ is a factor removed by hairpin deletion that determines the first transition that is taken. For each $j\in[|u|+1]\setminus\{1\}$, the factor $v_{j,\ell}$ is a factor that can be completely removed by haipin-deletion that brings us to the beginning of some factor $s_{i}$, $i\in[|Q|]$, from which point on we can now select the next transition. The factor $v_{j,r}$, for $j\in[|u|]\setminus\{1\}$, similar to $v_1$, again is a factor removed by hairpin deletion that determines the next transition that is taken. Finally, we continue with two last removed hairpins that result in transcription termination. As $q_{i_{|u|+1}}$ is a final state, the factor $v_{|u|+1,r}$ actually exists and can be removed by hairpin deletion. It is followed by $f_s$. Then, in a last step, we can remove $v_{|u|,r}$ to obtain a suffix $f_sf_e$ in the transcribed word that is followed by $tyw^\omega$, resulting in termination, i.e., resulting in $u\in[w^\omega]_{\hpdblogpmax_{S'}}$.

\textbf{Formal proof second direction ($[w^\omega]_{\hpdblogpmax_{S'}}\subseteq L(A)$):} Next, we need to show that we cannot obtain any word that is not in $L(A)$. Let $u\in[w^\omega]_{\hpdblogpmax_{S'}}$. By terminating hairpin deletion, we know that there exists a prefix $w_pt$ of $w^\omega$, for $t\in T$ and $w_p\in\Sigma^+$, such that $w_p \hpdbe us$, for some suffix $s\in\Sigma^m$ of length $m$ that forms a hairpin. For $w_p$, we know by the definition of parallel bounded hairpin deletion that there exists some $n\in\N$ such that $$w_p = u_1\ v_1\ ...\ u_n\ v_n\ u_{n+1}$$ with $u_i\in\Sigma^*$, for all $i\in[n+1]$, and $v_j\in\Sigma^+$, for all $j\in[n]$, for which we have $v_j\hpdbe\varepsilon$, and such that $us = u_1...u_{n+1}$. By the definition of maximal parallel bounded hairpin deletion, we know that for all $(x,y)\in C$ we have $x\notin\Fact(u_i)$, for all $i\in[n+1]$. For all $j\in[n]$, we know that $v_j$ starts with some left context $x$ for some $(x,y)\in C$.

First, we show inductively that we can only obtain letters that can be read during a computation of $A$. For that, we show that $u_1$ must be empty, that $v_1$ contains only a single deleted hairpin, that $u_2$ is just a single letter that represents a label of some outgoing transition of $q_1$ and that $v_2$ consists of exactly two subsequently removed hairpins (for simplicity reasons, in this proof, we assume that a single $v_i$ may contain multiple subsequently removed hairpins, also resulting in the empty word. Following the formal model, between each subsequently removed hairpin, a factor $u$ would have to be added, which is set to the empty word. In this proof, we always assume that such a factorization is possible if we talk about subsequently removed hairpins in a single factor $v$.) By induction, and using the same arguments used for the basic step, we obtain that all $u_i$, for $i\in[n-1]\setminus{1}$ must be single letters representing labels from transitions in $A$, that all $v_i$, with $i\in[n-1]$ being odd, contain just a single removed hairpin, and that all $v_j$, with $j\in[n-1]$ being even, contain exactly two subsequently removed hairpins (hence, could be split up to a factor containing two distinct hairpins $v_{j_1}$ and $v_{j_2}$ that surround the empty word in the above factorization). After that, we proof the terminating condition, resulting in $u$ being actually in $L(A)$.
    
Suppose $u_1\neq\varepsilon$. Then $v_1$ does not start at $w[1]$. But then, the next possible starting position for $v_1$ is the next occurring left context $x$ of some $(x,y)\in C$ that does not start in $w[1]$. By construction, this is the left context $S_{q_k}$ for some state $q_k\in Q$ that is located after the first encoded transition. However, if $v_1$ starts with that or a later occurring left context, then, e.g., $S_{1,1}\in\Fact(u_1)$, which is a contradiction as $(s_{1,1},E_{1,1})\in C$. So, $u_1=\varepsilon$ and $v_1$ has a prefix $S_{1,k}$ for some $\nth{k}$ transition of state $q_1$.

By construction, we know that there exists a unique right context $E_{1,k}$ such that $(S_{1,k},E_{1,k})\in C$. Suppose $v_1$ is factorized into $v_1 = v_{1,1}v_{1,2}$ (or more factors) such that, for both $i\in[2]$, we have $v_{1,i}\hpdbe\varepsilon$. Then, $v_{1,1} = S_{1,k}...E_{1,k}$ and $v_{1,2}$ would start immediately after $E_{1,k}$ in $w$. As before, $v_{1,2}$ would need to start with a left context $x$ for some $(x,y)\in C$. However, by construction of $w$ we know that after $E_{1,k}$ follows just the letter of the $\nth{k}$ transition of $q_1$, but no left context. So, this is a contradiction and we have $v = v_{1,1}$ as well as $u_2[1] = \ta_{1,k}$, where $\ta_{1,k}$ represents the letter of the $\nth{k}$ transition of $q_1$.

Now, suppose that $u_2$ has length $|u_2|>1$. Then $v_2$ cannot start with the left context $S_{\delta(q_1,\ta_{1_k})}$. But then, there are only $3$ possibilities for the next occurring left context in $w^\omega$. Either, it is the left context $S_{\delta(q_1,\ta_{1,k+1})}$, i.e., the left context representing the jump to another state for the $\nth{k+1}$ transition of $q_1$, or, if $q_1$ is also a final state, then the left context $S_{end}$ which is used in the termination process, or, if $q_1$ is not a final state and $q_1$ only has $k$ transitions, the collection of left contexts used to determine the choice of transitions in the state $q_2$, e.g., $S_{2,1}$, $S_{2,2}$, and so on. No matter the choice of the next left context, we observe that their occurrences are disjoint from the occurrence of $S_{\delta(q_1,\ta_{1,k})}$. Hence, $u_2$ has a prefix $\ta_{1,k}S_{\delta(q_1,\ta_{1,k})}$, which is a contradiction to the definition of maximal parallel bounded hairpin deletion. So, we can only have $u_1 = \ta_{1,k}$ and that $v_2$ has the prefix $S_{\delta(q_1,\ta_{1,j})}$. 
    
For readability purposes, until defined otherwise, from now on assume that $q_i = \delta(q_1,\ta_{i,j})$. Hence, $S_{\delta(q_1,\ta_{1,j})} = S_{q_i}$. In contrast to $v_1$, for $v_2$, we have to show that it always contains exactly $2$ subsequent hairpin deletion steps. First, as $v_2$ has the prefix $S_{q_i}$ for which only the unique right context $E_{q_i}$ in $(S_{q_i},E_{q_i})\in C$ exists, we know that $v_2$ has a prefix $S_{q_i}...E_{q_i}$.
By the definition of $w^\omega$, we know that immediately after that, we are in the beginning of the factor $s_i$. As $s_i$ is analogously constructed to $s_1$, we know by analogous arguments to the ones made for $u_1$ and $v_1$, that we have to start with some context $(S_{i,k'},E_{i,k'})\in C$, $k$ referring to the selected transition, to continue. This results in $v_2$ containing at least 2 subsequently removed hairpins. We know by the arguments from before that after the occurrence of $E_{i,k'}$ in $w^\omega$, there follows some letter $\ta_{i,k'}$ which represents the letter from the $\nth{k'}$ transition of $q_i$. Also, we know by the same arguments that this occurrence of $\ta_{i,k'}$ is not part of any left context in $C$. Hence, $v_2$ consists of exactly two subsequent hairpin deletion steps. 
    
In addition, we obtain that $u_3$ has the letter $\ta_{i,k'}$ as a prefix. By analogous arguments from before, we get that $u_2$, and in particular every following $u_i$, at least for $i\in[n-1]\setminus{1}$ (shown in terminating condition), is a single letter representing the letter of some transition of $A$. By induction, we also know that the order of those always represents a valid computation on $A$.

So, by now we know that $u$ must always have a prefix that represents the prefix of some word in $L(A)$. What's left to show is that termination can only occur, when $u$ actually is in $L(A)$. This can be done by a combinatorial argument on the construction of $w^\omega$.

First, we know that we can only terminate before the position of the factor $t\in T$ in $w$. Also, we know that we must need a hairpin stem of length $2m$ obtained by hairpin deletion immediately before that occurrence of $t$. By the construction of $s_e$ in $w$, we know that $t$ is preceded by the factor $f_e$ of length $m$. $f_e$ is not part of any context. Hence, $u_{n+1}$ has a suffix $f_e$. We must obtain the suffix $\theta(f_e)f_e$ by hairpin deletion from $w^\omega$. By construction, we know $\theta(f_s) = f_e$, so $\theta(f_e) = f_s$. By construction, we also know that $f_e$ is preceded by $E_{end}$ which is, by assumption, not a factor of $f_s$. Hence, $E_{end}$ needs to be removed by hairpin deletion, resulting in $u_{n+1} = f_e$. This is only possible with the corresponding context $(S_{end},E_{end})\in C$ which can only be found in the encoding words $s_i$ of final states $q_i\in F$. By construction, these occurrences of $S_{end}$ are preceded by $f_s$. Also by construction, $f_s$ is preceded by some right context $E_{i,k+1}$, assuming the encoding $s_i$ of some final state $q_i$ with $k$ transitions. Using the inductive arguments from before, we know that we can only obtain $f_s$ as part of some factor $u_n$ if and only if we select the context $(S_{i,1}...S_{i,k+1},E_{i,k+1})$ as the last removed hairpin in $v_{n-1}$. Hence, $u_n = f_s$.
    
Also by the inductive arguments from before, we know that $v_{n-1}$ consists of exactly two subsequently removed hairpins, where the first one is enclosed by $(S_{q_i},E_{q_i})$. For all other previous $u_j$, for $j\in[n-1]$, we know that they are all single letters representing a prefix of a valid computation on $A$. As we can obtain $f_s$ only in final states, we know that $u_1...u_{n-1}\in L(A)$. Also, we know that the suffix of length $2m$ is actually $s = f_sf_e = u_nu_{n+1}$. So, $u = u_1...u_{n-1}$ and by that $u\in L(A)$. This concludes this direction and the proof of this construction.
\newpage

\section{Additional Content Regarding the Example for Theorem~\ref{theorem:dfa-simulation} using the RNA Alphabet $\Sigma_{RNA} = \{\tA,\tU,\tC,\tG\}$}
\label{section:appendix-simulation-example}
\subsection{Detailed Examination of Factors in the General Construction}
\label{sec:appendix-simulation-example-general-examination}
This detailed examination is related to point $(4)$ in Section~\ref{sec_sim_rna-framework}.

For all contexts $(S_{i,1}...S_{i,j}, E_{i,j})\in C$, we see that they are bordered by either $\tA\tA$ or $\tU\tU$, followed and preceded by either $\tC$ or $\tG$. For all $S_{i,1}...S_{i,j}$, we know that after the first $\tA\tA$, there occur $i$ many $\tC$'s. We notice that there is only one position in the word where $\tA\tA(\tC)^{i}$ occurs and that is in the beginning of $s_i$. In addition to that, now regarding $E_{i,j}$, we notice that there are two positions where the corresponding $(\tG)^i\tU\tU$ occurs, and these are in either in the end of $\theta(S_{i,1}...S_{i,j})$ or the end of $E_{i,k}$. As we need an occurrence of a right context that is preceded by a valid hairpin, we see that the first occurrence has to be used in the hairpin-formation between $S_{i,1}...S_{i,j}$ and $\theta(S_{i,1}...S_{i,j})$, and we see that the second occurrence can then be used as a right context. So, no other unintended placements are possible. The same holds for each context $(S_{q_i},E_{q_i})\in C$ which are bordered by $\tA\tA\tA$ (resp. $\tU\tU\tU$), encapsulating a unique number of $o_i$ many $\tC$'s or $\tG$'s. As before, the word used in the right context also occurs two times, once for the hairpin formation and once for the right context recognition. The first occurrence has to be used to form a valid hairpin and the second one, as the margin is set to $0$, has to be a consecutive factor, i.e., the intended occurrence of $E_{q_i}$. The same can be said analogously for $(S_{end},E_{end})\in C$. Hence, this construction generally results in no unintended factors, as long as the numbers of $\tC$'s and $\tG$'s between the respective borders are pumped up enough. One has to make sure that the letters from the transitions $\ta_{i,k}$ are not part of some context. But this is prevented as these are always preceded by $\tU\tU$ in the end of $E_{i,k}$ and followed by $\tA\tA$ in $S_{q_{i,k}}$. So, they cannot be part of any context. The words $f_s$ and $f_e$ might also interfere with some context, but as both are bordered by $(\tA)^4$ (resp. $(\tU)^4$), this cannot happen. The termination word $t = (\tU)^10$ can also not be part of any context, as is is preceded by $(\tU)^4$ in $f_e$ and followed by $\tU\tU\tU$ in $\theta(S_{q_1})$ (see construction in proof of Theorem~\ref{theorem:dfa-simulation} to verify this).

\subsection{Specific Example Construction}
\label{sec:appendix-simulation-example-specific}
This section contains a specific example for the construction of a circular word $w^\omega$ for some NFA $A$ for which we have $L(A)=[w^\omega]_{\hpdblogpmax}$, using the framework given in Section~\ref{sec_sim_rna-framework}. Let $A = (Q,\Sigma_{\mathtt{RNA}},q_1,\delta,F)$ be some NFA that is defined by the following graph representation: 
\begin{figure}[h]
    \centering
    \includegraphics[width=4cm]{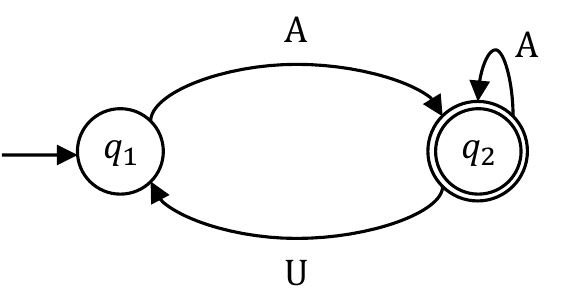}
    \caption{NFA $A$ with $Q = \{q_1,q_2\}$, $F = \{q_2\}$, and $\delta = \{((q_1,\tA),q_2), ((q_2,\tU),q_1), ((q_2,\tA),q_2)\}$. So it accepts the regular language $\tA^*(\tU\tA^*)^*$.}
    \label{fig:nfa-simulation-example}
\end{figure}
We use the construction in Theorem~\ref{theorem:dfa-simulation} and encoding given above to obtain a word $w$ for which we have $[w^\omega]_{\hpdblogpmax_S} = L(A)$ for the properties tuple $S = (\Sigma_{\mathtt{RNA}},\theta,1,C,0,\{t\},9)$. Notice that the margin is of length 0, that the set of terminating sequences only contains $t = (\tU)^{10}$, that the log factor is $1$, that the alphabet is $\Sigma_{\mathtt{RNA}}$, and that $\theta$ is given as above.

First, we set the state encodings $s_1$ and $s_2$. Notice that $q_1$ only has $1$ outgoing transition while $q_2$ has two outgoing transitions and is a final state. For $s_1$, we obtain
$$
s_1 =
\underbrace{\tA\tA\tC\tA\tA\tG\tA\tA}_{S_{1,1}}\ 
\underbrace{
\underbrace{\tU\tU\tC\tU\tU\tG\tU\tU}_{\theta(S_{1,1})}\ 
\underbrace{\tU\tU\tC\tU\tU\tG\tU\tU}_{E_{1,1}}\ 
\underbrace{\color{blue}\tA}_{\ta_{1,1}}\ 
\underbrace{\tA\tA\tA(\tG)^6\tA\tA\tA}_{S_{q_2}}\ 
}_{e_{1,1}}
\underbrace{\tU\tU\tU(\tG)^6\tU\tU\tU}_{\theta(S_{q_2})}\ 
\underbrace{\tU\tU\tU(\tG)^6\tU\tU\tU}_{E_{q_2}}.
$$
Notice that we set $o_2 = 6$ in $S_{q_2}$ and $E_{q_2}$ due to the total size of the word. For $s_2$, due to space constraints, we first give the abstract structure and then the assignment to each word. Notice that $s_2$ is the last state encoding before appending $s_e$ to $w$. Hence, no $\theta(S_{q_3})E_{q_3}$ has to be added in the end. We have
$$
s_2 =
S_{2,1}S_{2,2}S_{2,3}\ 
\underbrace{\theta(S_{2,1})E_{2,1}{\color{blue}\tU} S_{q_1}}_{e_1}\ 
\underbrace{\theta(S_{2,1}S_{2,2})E_{2,2}{\color{blue}\tA} S_{q_2}}_{e_2}\ 
\underbrace{\theta(S_{2,1}S_{2,2}S_{2,3})E_{2,3}{\color{orange}f_s} S_{end}}_{e_3}\ 
$$
for the following assignments:
$$
S_{2,1}S_{2,2}S_{2,3} = 
\underbrace{\tA\tA\tC\tC\tA\tA\tG\tA\tA}_{S_{1,1}}
\underbrace{\tG\tA\tA}_{S_{2,2}}
\underbrace{\tG\tA\tA}_{S_{2,3}}
$$
$$\theta(S_{2,1})E_{2,1}\ {\color{blue}\tU}\ S_{q_1} = 
\underbrace{\tU\tU\tC\tU\tU\tG\tU\tU}_{\theta(S_{2,1})}\ 
\underbrace{\tU\tU\tC\tU\tU\tG\tU\tU}_{E_{2,1}}\ 
{\color{blue}\tU}\ 
\underbrace{\tA\tA\tA(\tC)^5\tA\tA\tA}_{S_{q_1}}
$$
$$
\theta(S_{2,1}S_{2,2})E_{2,2}\ {\color{blue}\tA}\ S_{q_2} =
\underbrace{\tU\tU\tC\tU\tU\tC\tU\tU\tG\tU\tU}_{\theta(S_{2,1}\theta(S_{2,2}))}\ 
\underbrace{\tU\tU\tC\tU\tU\tC\tU\tU\tG\tU\tU}_{E_{2,2}}\ 
{\color{blue}\tA}\ 
\underbrace{\tA\tA\tA(\tC)^6\tA\tA\tA}_{S_{q_2}}
$$
$$
\theta(S_{2,1}S_{2,2}S_{2,3})E_{2,3}f_s S_{end} =
\underbrace{\tU\tU\tC\tU\tU\tC\tU\tU\tC\tU\tU\tG\tU\tU}_{\theta(S_{2,1}S_{2,2}S_{2,3})}\ 
\underbrace{\tU\tU\tC\tU\tU\tC\tU\tU\tC\tU\tU\tG\tU\tU}_{E_{2,3}}\ 
\underbrace{\color{orange}(\tA)^4\tG(\tA)^4}_{\color{orange}f_s}\ 
\underbrace{\tA\tA\tA(\tC)^7\tA\tA\tA}_{S_{end}}
$$
Notice that we set, in addition to $o_2 = 6$, the numbers $o_1 = 5$ and $o_e = 7$ for the words $S_{q_1},E_{q_1}$, $S_{end}$ (thus also for $E_{end}$ in $s_e$). The numbers are selected such that hairpin deletion regarding the logarithmic model works. Finally, $s_e$ is encoded in the following manner:
$$
s_e = 
\underbrace{\tU\tU\tU(\tG)^7\tU\tU\tU}_{\theta(S_{end})}\ 
\underbrace{\tU\tU\tU(\tG)^7\tU\tU\tU}_{E_{end}}\ 
\underbrace{\color{orange}(\tU)^4\tC(\tU)^4}_{\color{orange}f_e}\ 
\underbrace{\color{red}\tU\tU\tU\tU\tU\tU\tU\tU\tU\tU}_{\color{red}t}\ 
\underbrace{\tU\tU\tU(\tG)^5\tU\tU\tU}_{\theta(S_{q_1})}\ 
\underbrace{\tU\tU\tU(\tG)^5\tU\tU\tU}_{E_{q_1}}
$$
Remember, that $\theta(S_{q_1})E_{q_1}$ is needed to jump to $s_1$. It is encoded in the end of $s_e$ such that each word starts with the left context $S_{1,1}$ in $s_1$. In total, we set
$$w = s_1\ s_2\ s_e.$$
We have to check that each intended hairpin can actually be formed, regarding the length bound obtained by the logarithmic energy model. But, for that we notice that each left context for the transition selection has at least length $|S_{1,1}| = 8$, which supports a loop of length $2^8=256$, which clearly works out inside of each $s_1$ and $s_2$. Regarding jumps between $s_1$, $s_2$ and $s_e$, notice that the shortest context responsible for such a jump, $S_{q_1}$, has length $|S_{q_1}| = 11$, which supports a loop of length $2^{11} = 2048$. That is longer that $w$ itself, hence such a log-hairpin can also always be formed.

We obtain $[w]_{\hpdblogpmax_S} = L(A)$ by the proof of Theorem–\ref{theorem:dfa-simulation} and the fact, that no factor occurs in an unintended position.

\newpage

\section{Additional Content Regarding Section~\ref{sec_dec}}
\label{section:appendix-proofs-ssbnfas}
\subsection{Formal Definitions of RB-NFAs and DB-NFAs}

\begin{definition}
    A \emph{return-bounded nondeterministic finite automaton} (RB-NFA) $A$ is a hexatuple $A = (Q,c,\Sigma,\delta,q_0,F)$
    where w.l.o.g. $Q = \{q_1,...,q_{|Q|}\}$ is a \emph{ordered finite set of states} with $q_i < q_{i+1}$ for $i\in[|Q|-1]$, $c\in\N$ denotes the \emph{return-bound},
    $\Sigma$ denotes a \emph{finite alphabet}, $\delta : Q\times\Sigma \rightarrow Q$ denotes the transition function,
    $q_0\in Q$ the initial state, and $F\subseteq Q$ the set of final states.

    The language $L(A)$ of a RB-NFA is the language of all word $w\in\Sigma^*$ where $\delta(q_1,w)\in F$
    and the shortest accepting path $q_0 = p_{1}\rightarrow \cdots \rightarrow p_\ell = \delta(q_0,w)$ satisfies 
    the condition that there are no more than $c$ pairs $p_{i}> p_{i+1}$ for $i\in[\ell-1]$.
\end{definition}

\begin{definition}
    A \emph{distance-bounded nondeterministic finite automaton} (DB-NFA) $A$ is a heptatuple $A = (Q,f_d,c,\Sigma,\delta,q_0,F)$
    where w.l.o.g. $Q = \{q_1,...,q_{|Q|}\}$ is a \emph{ordered finite set of states} with $q_i < q_{i+1}$ for $i\in[|Q|-1]$, $c\in\N$ denotes the \emph{distance-bound},
    $\Sigma$ denotes a \emph{finite alphabet}, $\delta : Q\times\Sigma \rightarrow Q$ denotes the transition function,
    $q_0\in Q$ the initial state, and $F\subseteq Q$ the set of final states, and $f_w : Q\times Q \rightarrow \N$ is a distance function such that $f_w(q_i,q_{i+1})\in\N$ for $i\in[|Q|-1]$, $f_w(q_{|Q|},q_1)\in\N$, $f_w(q_i,q_j) = \sum_{i'\in[i..j-1]}f_w(q_{i'},q_{i'+1})$ if $i < j-1$, and $f_w(q_i,q_j) = f_w(q_i,q_{|Q|}) + f_w(q_{|Q|},q_1) + f_w(q_1,q_j)$ if $j < i$.

    The language $L(A)$ of a DB-NFA is the language of all word $w\in\Sigma^*$ where $\delta(q_1,w)\in F$
    and there exists an accepting path $q_1 = q_{i_1}\rightarrow \cdots \rightarrow q_\ell = \delta(q_1,w)$ satisfies 
    the condition that $f_w(q_{i_j},q_{i_{j+1}}) \leq c$ for $j\in[\ell-1]$.
\end{definition}

\subsection{Omitted Proofs of Section~\ref{sec_dec}}

Proof of \textbf{Proposition~\ref{prop:hardness-small-cover-set-avoiding-forbidden-patterns-problem}}
\begin{proof}
    We can reduce the minimisation of nondeterministic finite cover automata (NFCA), which is known to be NP-hard by a result of Campeanu's~\cite{Campeanu2015}, to Problem~\ref{problem:small-cover-set-avoiding-forbidden-patterns}(cover, states). Let $L\subset\Sigma^*$ be a finite regular language and let $\ell\in\N$ such that for the longest word $w_j\in L$ we have $|w_j| = \ell$. Let $A = (Q,\Sigma,q_0,\delta,F)$ be some NFA. We reduce the problem of checking whether $A$ is a minimal NFCA for $L$ to the problem of checking whether some NFA is minimal regarding a set of target words $W$ and a set of forbidden patterns $F$. We leave the NFA $A$ unchanged. Let $W = L$ and $F = \Sigma$. Then $(W,F,A)$ is a valid instance for Problem~\ref{problem:small-cover-set-avoiding-forbidden-patterns}(cover, states).

    Assume $A$ is a state-minimal NFCA for $L$, i.e. $L(A) \cap \Sigma^{\leq\ell} = L$.
    Suppose $A$ is not a state-minimal NFA such that 
    $W\subseteq L_{\leq \ell}(A)$ and 
    $\left( L_{\leq \ell}(A)\setminus W \right) \cap \Sigma^* F\Sigma^* = \emptyset$.
    Then, there exists some smaller NFA $B = (Q',\Sigma,q_0,\delta,F)$ such that $|Q'| < |Q|$ with
    $W\subseteq L_{\leq \ell}(B)$ and $\left( L_{\leq \ell}(B)\setminus W \right) \cap \Sigma^* F\Sigma^* = \emptyset$.
    Then, we directly have that for all $u\in L$ we also have $u\in L(B)$ as $u\in W$.
    We also have for all $u'\in\Sigma^{\leq \ell}$ with $u'\notin W$ that $u'\notin L(B)$ as $F = \Sigma$ and at least some letter
    of the alphabet always occur sin $u'$. Hence, $B$ is a smaller NFCA for $L$ which is a contradiction to the assumption.
    So, $A$ must be state-minimal for Problem \ref{problem:small-cover-set-avoiding-forbidden-patterns}(cover, states).

    Now, assume $A$ is not a minimal NFCA for $L$. Then there exists some smaller NFCA $B = (Q',\Sigma,q_0,F)$
    such that $|Q'| < |Q|$ and $L(B) \cap \Sigma^{\leq\ell} = L$.
    Let $u\in W$. Then $u\in L(B)$ as $B$ is a NFCA for $L$. By that we have $W \subseteq L_{\leq\ell}(B)$.
    Let $u'\in\Sigma^{\leq\ell}$ such that $u'\notin W$.
    As $B$ is a NFCA for $L$, we get that $u'\notin L(B)$.
    This results in $\left( L_{\leq \ell}(B)\setminus W \right) \cap \Sigma^* F\Sigma^* = \emptyset$ as 
    $L_{\leq \ell}(B)\setminus W = \emptyset$.
    Hence, we have that $A$ is not a minimal NFA regarding $W$ and $F$ for Problem \ref{problem:small-cover-set-avoiding-forbidden-patterns}(cover, states).

    By the above we get that solving a $(W,\Sigma,A)$ instance of Problem 2(cover, states) is equivalent to finding the minimal NFCA accepting $L_{\leq\ell}(A)$ for some input NFA $A$.
\end{proof}

Proof of \textbf{Proposition~\ref{lemma:ssb-nfa-min-in-np}}
\begin{proof}

    \textbf{When $L$ is given as a set:}
    We show it for $\mathtt{SSB\text{-}NFA\text{-}Min\text{-}States}$. The proof for $\mathtt{SSB\text{-}NFA\text{-}Min\text{-}Transitions}$
    works analogously.
    Let $L$ be any regular language and let $k,c\in\N$ be numbers.
    We know a trivial automaton for $L$ has a single path for each word in $L$.
    This trivial automaton has at most $\sum_{w\in L}(|w|) +1$ many states and transitions which is polynomial in the size of the input as $L$ is given plainly as a set of strings. So, if $k \geq \sum_{w\in L}(|w|) +1$, the answer is always $\mathtt{true}$.
    So, assume $k < \sum_{w\in L}(|w|) +1$. Guess some SSB-NFA $A = (Q,c',\Sigma,\delta,q_0,F)$ with $|Q| \leq k$ and $c' = c$.
    Check whether $L(A) = L$. Checking $L \subseteq L(A)$ can be done in linear time by checking whether for all $w\in L$ we have $w\in L(A)$. To check whether there exists some word $w\in L(A)$ for which we have $w\notin L$, we can do the following. Remove each state (and the corrsponding transitions) which do not lead to a final state. This can be done in $O(k)$ by reversing the direction of all transitions, adding a new initial state which points to all final states, and check with BFS/DFS which states are reachable from that new initial state.
    Now, we construct the prefix-tree $T_L$ of $L$. This can be done in polynomial time and space. We know the max depth of $T_L$ is the length of the longest word in $L$.
    Now, we iteratively generate the prefix tree $T_{L(A),k}$ of all words of length $k$ contained in $L(A)$ using BFS, starting with length $1$ and moving up. In each iteration, we check whether $T_{L(A),k}$ is completely contained in $T_L$.
    If that is not the case, we know that $L(A)$ has a word which is not in $L$ by the assumption that no path leads to a non-final state. We know that $L(A) \subseteq L$ if after the first iteration where $T_{L(A),k} = T_{L(A),k+1}$ we have that $T_{L(A),k+1}$ is completely contained in $T_{L}$. We know this whole procedure takes polynomial time as we have at most
    $|w|$ many iterations (for $w\in L$ being the longest word in $L$) and we either have to accept or reject at this point and, additionally, for each iteration we know that $T_{L(A),k}$ is of smaller or equal size than $T_L$ as otherwise we would have rejected already, hence comparing $T_L$ and $T_{L(A),k}$ is polynomially upper bounded by the size of $T_L$. By that we get that we can check whether $L(A) \subseteq L$ in polynomial time.

    \textbf{When $L$ is represented by some DFA, NFA, or regular expression:}
    Again, the show it just for $\mathtt{SSB\text{-}NFA\text{-}Min\text{-}States}$. 
    The proof for the transition variant works analogously.
    Given any regular language $L$ and numbers $k,c\in\N$, first check if its finite. If not, return $\mathtt{false}$.
    Depending on the representation of $L$, different settings may occur.
    First, if $L$ is given as a DFA $A_d = (Q_d, \Sigma, \delta_d, q_{d,0}, F_d)$, we know that no loops can occur on paths leading to a final state in $A_d$ as otherwise $L$ would be an infinite language. If $|Q_d| \leq k$, then $A_d$ is immediately a solution for the problem and we can return $\mathtt{true}$. 
    If $|Q_d| > k$, guess some SSB-NFA $A = (Q,c',\Sigma,\delta,q_0,F)$ such that $|Q|\leq k$ and $c' = c$
    and check whether $L(A) = L$. As $|Q| \leq k < |Q_d|$, we know that we only need a polynomial amount of space regarding the input space.
    If $L$ is given as a NFA, the same holds analogously.
    If $L$ is given as a regular expression,
    transfer it to a NFA $A' = (Q',\Sigma,\delta',q_0',F')$ in polynomial time~\cite{hopcroft}.
    By the polynomial time constraint to convert the regular expression to the NFA $A'$, we know that $A'$ is in polynomial space regarding the input. Now use the same arguments as before in an analogous manner.
    This concludes the claim.
\end{proof}

Proof of \textbf{Proposition~\ref{lemma:ssb-nfa-min-hard}}.
\begin{proof}
    We begin with a proof that assumes an unbounded alphabet size to obtain an intuition. After that, we give a construction what works for any bounded alphabet.

    \textbf{Unbounded alphabet size:} The idea of this proof is very similar to the ones from \cite{DBLP:conf/lata/GruberH07} and \cite{Amilhastre2001}. The result is shown by reducing the $\mathtt{Biclique\text{-}Cover\text{-}Number}$ problem to $\mathtt{SSB\text{-}NFA\text{-}Min}$. The first one is known to be NP-complete \cite{Orlin1977}. Let $G = (V,E)$ be some bipartite graph such that $V = V_1 \cup V_2$ with $V_1 = \{a_1, a_2, ... , a_n\}$ and $V_2 = \{b_1, b_2, ... , b_m\}$ such that for all $(u,v)\in E$ we have $u\in V_1$ and $v\in V_2$. Let $k\in\N$ be some number. We construct an alphabet $\Sigma = V$ and a language $L = \{\ a_ib_j\ |\ (a_i,b_j)\in E\}$ over $\Sigma$. Additionally, we set $c = 1$ and $k' = k+2$.

    For the first direction, assume there exists a covering of $C = \{C_1,...,C_{k_c}\}$ containing $k_c \leq k$ many bicliques. Then, we construct a SSB-NFA $A = (Q,c,\Sigma,\delta,q_s,F)$ such that $Q = \{q_s, q_1, ... , q_{k_c}, q_f\}$, and $F = \{q_f\}$. Clearly $|Q| = k_c+2 \leq k'$. For each $C_i \in C$ we have $C_i = (V_i,E_i)$ for some $V_i\subseteq V$ and $E_i\subseteq E$ with $V_i = V_{i,1} \cup V_{i,2}$ such that for all $(a_{j_1},b_{j_2})\in E_i$ we have $a_{j_1}\in V_{i,1}$ and $b_{j_2}\in V_{i,2}$. For each such $(a_{j_1},b_{j_2})\in E_i$ we add the transitions $\delta(q_s,a_{j_1}) = q_i$ and $\delta(q_i,b_{j_2}) = q_f$ to $A$. This results in $a_{j_1}b_{j_2}\in L(A)$ for all $a_{j_1},b_{j_2}\in V_i$ as $C_i$ is a biclique. As $C$ covers the whole graph $G$ we know that $L \subseteq L(A)$. As no other transition to the final state are added, we do not produce any additional words. Hence $L(A) = L$ which concludes this direction.

    For the other direction, assume there exists a SSB-NFA $A = (Q,c',\Sigma,\delta,q_s,F)$ such that $|Q| \leq k' = k + 2$ and $c' \leq c = 1$ (hence $c' = 1$) such that $L(A) = L$. As $c = 1$, no state can be reached from itself again to obtain any words in $L$. So, there exists at least one final state $q_f \in Q$ with $q_f \neq q_s$. This leaves us with potentially $k$ additional states. As all words in $L$ are of length $2$ and no state can reach itself from itself again, we know that for each word $w = \in L$ with $w = a_{j_1}b_{j_2}$ there exists a path of $3$ states $q_sq_iq_f$ in $A$ where $q_f\in F$ and $q_i$ is an intermediate state not equal to $q_s$ or $q_f$ such that $\delta(q_s,a_{j_1}) = q_i$ and $\delta(s_f,a_{j_2})$. Construct for each such intermediate state $q_i$ a graph $C_i = (V_i,E_i)$ such that $V_i = V_{i,1} \cup V_{i,2}$ with $a_{j_1}\in V_{i,1}$, $b_{j_2}\in V_{i,2}$, and $(a_{j_1},b_{j_2})\in E_i$ for all $a_{j_1},b_{j_2\in\Sigma}$ with $\delta(q_s,a_{j_1}) = q_i$ and $\delta(s_f,a_{j_2})$. By this construction, we see that $C_i$ must be a biclique. We know that $A$ has at most $k$ such intermediate steps to obtain $L$. Hence, we obtain a set $C$ of at most $k$ bicliques which cover the whole graph $G$ by the way $L$ is constructed. This concludes this direction.

    For the transition variant of $\mathtt{SSB\text{-}NFA\text{-}Min}$ we also obtain NP-hardness by the same reduction by adapting the variable $k'$ such that $k' = 2k$. The arguments are analogous to the case related to state minimization.

    \textbf{Bounded alphabet size:} We show the binary case out of which we can reduce to any alphabet size greater than $2$.
    Let $G=(V_1,V_2,E)$ be a bipartite graph and let $k\in\N$ be a number.
    W.l.o.g. assume that $V_1 = \{a_1, ... , a_n\}$ and $V_2 = \{b_1, ... , b_m\}$.
    Set $\Sigma := \{0,1\}$ and construct a language $L$ such that
    $ L := \{\ 0^i110^j\ |\ (a_i,b_j)\in E \text{ for } i\in[n],j\in[m]\}. $
    Also set $k' = 2+n+m+k$ and $c=1$.

    For the first direction, assume there exists a biclique covering $C := \{C_1,...,C_{k_c}\}$ of $G$ with $k_c \leq k$
    many biclques such that for each $r\in[k_c]$ we have $C_r = (V_{1,r},V_{2,r},E_r)$. 
    Let $A = (Q,c,\Sigma,\delta,q_s,F)$ be a SSB-NFA such that
    $$Q := \{\ q_s,q_f\ \} \cup \{\ q_{a_i}\ |\ i\in[n]\ \} \cup \{\ q_{b_j}\ |\ j\in[m]\ \} \cup \{\ q_{c_r}\ |\ r\in[k_c]\ \},$$
    $F := \{q_f\}$ and $\delta : Q\times\Sigma \rightarrow Q$ is defined by
    \begin{itemize}
        \item $\delta(q_s,0) := \{ q_{a_1} \}$,
        \item $\delta(q_{a_i},0) := \{ q_{a_{i+1}} \}$ for $i\in[n-1]$,
        \item $\delta(q_{b_1},0) := \{ q_f \}$,
        \item $\delta(q_{b_i},0) := \{ b_{b_{i-1}} \}$ for $i\in[m]\setminus\{1\}$, and
        \item $q_{c_r}\in\delta(q_{a_i},1)$ and $q_{b_j}\in\delta(q_{c_r},1)$ if $(a_i,b_j)\in E_r$ f.s. $r\in[k_c]$.
    \end{itemize}
    Then $|Q| = 2+n+m+k_c \leq 2+n+m+k = k'$ (see Figure \ref{fig:ssb-nfa-min-hardness-sketch}). 
    Now, we have to show that $L = L(A)$.
    Let $w\in L$. Then $w=0^i110^j$ for some $i\in[n]$ and $j\in[m]$.
    By definition, there exists some edge $(a_i,b_j)\in E$ and, in extension, as $C$ is a biclique covering,
    there also exists a biclique $C_r\in C$ such that $(a_i,b_j)\in C_r$ for some $r\in[k_c]$.
    So, by construction, $q_{a_i}\in\delta(q_s,0^i)$, $q_{c_r}\in\delta(q_{a_i},1)$, $q_{b_j}\in\delta(q_{c_r},1)$,
    and $q_f\in\delta(q_{b_j},0^j)$. Hence $0^i110^j\in L(A)$. 
    Now let $w\in L(A)$. We know $q_f$ is the only final state in $A$. To reach $q_f$, there must exist
    $i\in[n]$, $j\in[m]$, and $r\in[k_c]$ such that $q_f\in\delta(q_{b_j},0^j)$, $q_{b_j}\in\delta(q_{c_r},1)$,
    $q_{c_r}\in\delta(q_{a_i},1)$, and $q_{a_i}\in\delta(q_s,0^i)$. Hence, $w = 0^i110^j$.
    By construction, $q_{b_j}\in\delta(q_{c_r},1)$ and $q_{c_r}\in\delta(q_{a_i},1)$ only if $(a_i,b_j)\in E_r$.
    By that, we also know $(a_i,b_j)\in E$ out of which $w\in L$ follows by construction. So, we get $L(A) = L$
    which concludes this direction.

    \begin{figure}[h]
    \centering
    \includegraphics[width=13cm]{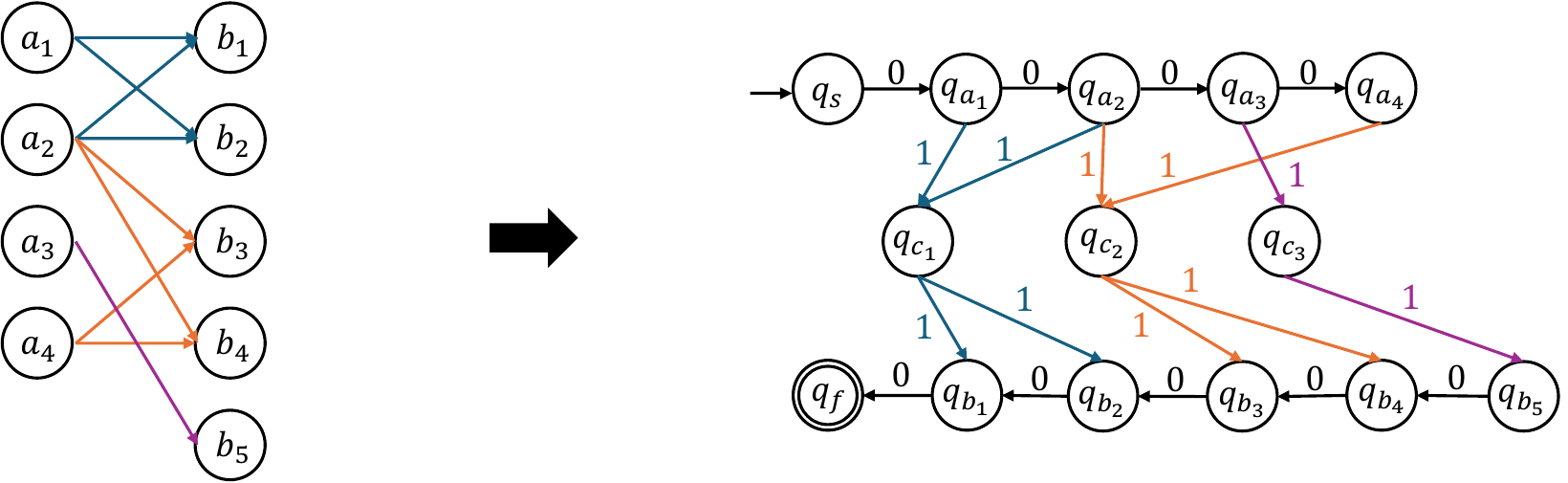}
    \caption{Sketch of the construction used in Proposition \ref{lemma:ssb-nfa-min-hard}. The left side shows a bipartite graph for which a biclique cover with $k=3$ bicliques exists (blue, orange, purple). The right side shows an automaton recognizing the constructed language $L$ (see proof of Proposition~ \ref{lemma:ssb-nfa-min-hard}) using $2+n+m+k$ many states. Each of the $3$ states in the middle poses as one of the cliques.}
    \label{fig:ssb-nfa-min-hardness-sketch}
    \end{figure}

    For the other direction assume there exists some SSB-NFA $A = (Q,c',\Sigma,\delta,q_s,F)$ such that $L = L(A)$,
    $|Q| \leq k' = 2+n+m+k$, and $c' = c = 1$. As $L(A) = L$ we know
    that $0^n110^j,0^i110^m\in L(A)$ for some $j\in[m]$ and $i\in[n]$. As $c'=1$, there must exist a chain of $n$ many
    states after the initial state (w.l.o.g. called $q_{a_1}$ to $q_{a_n}$) such that $q_{a_1}\in \delta(q_s,0)$
    and $q_{a_{i+1}}\in\delta(q_{a_i},0)$ for $i\in[n-1]$ as well as a chain of $m$ many states (w.l.o.g. called $q_{b_1}$ to $q_{b_m}$) leading to a final state $q_f$ such that $q_f\in\delta(q_{b_{1}},0)$ and $q_{b_{i-1}}\in\delta(q_{b_i,0})$ for
    $i\in[m]\setminus\{1\}$. All of the aforementioned states must be pairwise distinct as, first, no state may be reached more than once, every word in $L(A)$ has a factor $11$ and if some state of $q_s,q_{a_1},...,q_{a_n}$ is equal to some state of
    $q_{b_{1}},...,q_{b_m},q_f$ there would exist a word $u\in L(A)$ such that $u\in\{0\}^*$ which is a contradiction.
    So we know of the existence of $n+m+2$ many pairwise distinct states which leaves us with up to $k$ many states.
    Let $w\in L(A)$. Then $w=0^i110^j$ for some $i\in[n]$ and $j\in[m]$. Then there exists a state $q_{c_r}\in Q$
    such that $q_{c_r}\in\delta(q_s,0^i1)$ and $\delta(q_{c,r},10^j)\in F$. Clearly, $q_{c_r} \neq q_s$
    and $q_{c_r}\notin F$ as $c'=1$ and otherwise $0^i1\in F$ which is a contradiction to the definition of $L$.
    Suppose $q_{c_r} = q_{a_{i'}}$ for some $i'\in[n]$. As $c' = 1$ we can assume $i' > i$.
    We know $0^n110^{j'}\in L$ for some $j'\in[m]$ and additionally $q_{a_{i'}}\in\delta(q_s,0^{a_{i'}})$
    and $q_{a_n}\in\delta(q_{a_{i'}},0^{n-i'})$.
    By that, we have that $0^i10^{n-i'}110^{j'}\in L$ which is a contradiction.
    Now, suppose $q_{c_r} = q_{b_{j'}}$ for some $j'\in[m]$. Then we have that $0^i10^{j'}\in L$ which is a contradiction.
    So, $q_{c_r}$ is a distinct state and one of the up to $k$ states which were left. From now on, we call each of those states $q_{c_r}$ an intermediate state.
    For each of those up to $k$ intermediate states $q_{c_r}$ we construct a graph $C_r = (v_{1,r},V_{2,r},E_r)$
    such that for all $w = 0^i110^j\in L(A)$ we have that if $q_{c_r}\in\delta(q_s,0^i1)$ and $\delta(q_{c_r},10^j)\in F$
    then $a_i\in V_{1,r}$, $b_j\in V_{2,r}$, and $(a_i,b_j)\in E_r$. Clearly, $C_r$ is a full bipartite graph (i.e. a biclique).
    As every word in $L(A)$ uses such an intermediate state and there are up to $k$ such intermediate states, we know that we
    construct up to $k$ such bicliques $C_r$ and that every edge in $E$ occurs in at least one of those bicliques $C_r$.
    By that, we have found a biclique covering of $G$ with up to $k$ bicliques.
    This concludes this direction and the claim overall.

    This result also follows for the transition variant of $\mathtt{SSB\text{-}NFA\text{-}Min}$ by using an analogous reduction, but only setting $k' = m+n+2k$. The obtained automata are the same, we are just counting the number of transitions instead of the number of states. The argument works the same.
\end{proof}

Sketch how to proof \textbf{Proposition~\ref{proposition:rb-nfa-min-np-hard}}
\begin{proof}
    \textbf{RB-NFAs:}
    Like in the bounded alphabet case of Proposition \ref{lemma:ssb-nfa-min-hard}, we set $k'$ and $c$ in exactly the same way.
    For the first direction, assume an order $<_Q$ on the states
    such that for all $i\in[n-1]$, $j\in[m]\setminus{1}$, and $r\in[k_c-1]$ we have $q_s <_Q q_{a_1}$,
    $q_{a_i} <_Q q_{a_{i+1}}$, $q_{a_n} <_Q q_{c_1}$, $q_{c_r} <_Q q_{c_{r+1}}$, $q_{c_{k_c}} <_Q q_{b_m}$,
    $q_{b_j} <_Q q_{b_{j-1}}$, and $q_{b_1} < q_f$. Then the result follows.
    The proof for the other direction works by exactly the same arguments given as in the bounded alphabet case of Proposition $\ref{lemma:ssb-nfa-min-hard}$
    and analogously applying the argument given in this proof for the first direction.

    \textbf{DB-NFAs:}
    Like in the bounded alphabet case of Proposition \ref{lemma:ssb-nfa-min-hard}, we set $k'$ exactly the same. Assume $|V_i| \geq |V_j|$ for $i,j\in{1,2}$ with $i \neq j$. Then set $c = k'$ as the maximum distance between to following states in a computation. 
    For the first direction, assume the same order $<_Q$ as in the proof of Proposition \ref{proposition:rb-nfa-min-np-hard}.
    Also set $f_w(q_i,q_{i+1}) = 1$ and $f_w(q_{|Q|},q_{1}) = 1.$
    Then notice that the largest distances between states that are possible are either the transition between $q_{a_1}$ and $q_{c_{k_c}}$ or the transition between $q_{c_1}$ and $q_{b_1}$. By the way we defined $c$, we don't surpass that value.
    For the other direction, notice that we can always set $f_w$ as above (i.e. setting all subsequent distances to 1) and set $c \leq k'$
    and then apply analogous arguments to the proofs of Proposition $\ref{proposition:rb-nfa-min-np-hard}$ and the bounded alphabet case Proposition $\ref{lemma:ssb-nfa-min-hard}$.
\end{proof}

\end{document}